\documentclass[final,5p,times,twocolumn]{elsarticle}

\makeatletter
\def\ps@pprintTitle{%
 \let\@oddhead\@empty
 \let\@evenhead\@empty
 \def\@oddfoot{}%
 \let\@evenfoot\@oddfoot}
\makeatother

\biboptions{sort&compress}

\usepackage{amsthm,amssymb,amsmath,amsfonts,amstext,bbm,mathtools}
\usepackage[mathscr]{eucal}
\usepackage{multicol}
\usepackage{graphicx,psfrag,float,epsfig,subfigure}
\usepackage{stmaryrd}
\usepackage[english]{babel}
\usepackage{url}

\usepackage{dsfont}
\usepackage{array}
\usepackage{booktabs}
\usepackage{makecell}

 \usepackage[monochrome]{color}

\usepackage[shortlabels]{enumitem}

\usepackage{epsfig}
\usepackage[linesnumbered,ruled,vlined]{algorithm2e}
\SetKwInput{KwInput}{Input}                %
\SetKwInput{KwOutput}{Output}              %
\newcommand{\nonl}{\renewcommand{\nl}{\let\nl\oldnl}}%

\theoremstyle{plain}
\newtheorem{theorem}{Theorem}

\newtheorem{proposition}{Proposition}

\theoremstyle{definition}
\newtheorem{example}{Example}
\newtheorem{assumption}{Assumption}
\newtheorem{definition}{Definition}

\theoremstyle{remark}
\newtheorem{remark}{Remark}

\newenvironment{pfof}[1]{\vspace{1ex}\noindent{\itshape Proof of #1:}\hspace{0.5em}}{\hfill\qed\vspace{1ex}}

\newcommand\xqed[1]{%
  \leavevmode\unskip\penalty9999 \hbox{}\nobreak\hfill
  \quad\hbox{#1}}
\newcommand\EXND{\xqed{$\Diamond$}}

\DeclareMathOperator{\col}{col}
\DeclareMathOperator{\im}{im}

\newcommand{\norm}[1]{\ensuremath{\left\| #1 \right\|}}

\newcommand{\calN}{\ensuremath{\mathcal{N}}}
\newcommand{\calI}{\ensuremath{\mathcal{I}}}
\newcommand{\calX}{\ensuremath{\mathcal{X}}}
\newcommand{\calV}{\ensuremath{\mathcal{V}}}
\newcommand{\calW}{\ensuremath{\mathcal{W}}}
\newcommand{\calO}{\ensuremath{\mathcal{O}}}
\newcommand{\ENC}[1]{\ensuremath{{\color{blue}\mathcal{E}}({#1})}}

\newcommand{\scrN}{\ensuremath{\mathscr{M}}}
\newcommand{\scrP}{\ensuremath{\mathscr{P}}}
\newcommand{\scrQ}{\ensuremath{\mathscr{Q}}}

\newcommand{\ONI}{\ensuremath{\overline{\calN}_i}}
\newcommand{\calP}{\ensuremath{\mathcal{P}}}
\newcommand{\R}{\ensuremath{\mathbb R}}
\newcommand{\Z}{\ensuremath{\mathbb Z}}

\newcommand{\Gen}{\fontfamily{phv}\selectfont Gen}
\newcommand{\Enc}{\fontfamily{phv}\selectfont Enc}
\newcommand{\Dec}{\fontfamily{phv}\selectfont Dec}

\newcommand{\View}{\text{\fontfamily{phv}\selectfont View}}
\newcommand{\Sim}{\text{\fontfamily{phv}\selectfont Sim}}

\newcommand{\PV}{\text{\fontfamily{phv}\selectfont Pv}}

\newcommand{\qj}{\ensuremath{q_j}}
\newcommand{\p}{\ensuremath{p_i}}
\newcommand{\pj}{\ensuremath{p_j}}

\newcommand{\bse}{\begin{subequations}}
\newcommand{\ese}{\end{subequations}}
\def\be{\begin{equation}}
\def\ee{\end{equation}}

\newcommand{\bbm}{\begin{bmatrix}}
\newcommand{\ebm}{\end{bmatrix}}

\newcommand{\sa}{\ensuremath{\boldsymbol{a}}}
\newcommand{\sm}{\ensuremath{\boldsymbol{m}}}
\newcommand{\Fa}{\ensuremath{\boldsymbol{\alpha}}}
\newcommand{\Fm}{\ensuremath{\boldsymbol{\kappa}}}

\newcommand{\ski}{\ensuremath{\text{sk}_i}}
\newcommand{\x}{\ensuremath{\boldsymbol{x}}}
\newcommand{\bld}[1]{\ensuremath{\boldsymbol{#1}}}

\usepackage[prependcaption,colorinlistoftodos]{todonotes}

\usepackage[normalem]{ulem}

\newcommand{\blue}[1]{\textcolor{blue}{#1}}

\begin{document}

\begin{frontmatter}

\title{Private Computation of Polynomials over Networks}

\author[mymainaddress]{Teimour Hossienalizadeh\corref{mycorrespondingauthor}}
\cortext[mycorrespondingauthor]{Corresponding author.}
\ead{t.hosseinalizadeh@rug.nl}

\author[mysecondaryaddress]{Fatih Turkmen}
\ead{f.turkmen@rug.nl}

\author[mymainaddress]{Nima Monshizadeh}
\ead{n.monshizadeh@rug.nl}

\address[mymainaddress]{Engineering and Technology Institute, University of Groningen, The Netherlands}
\address[mysecondaryaddress]{Bernoulli Institute for Mathematics, Computer Science and Artificial Intelligence, University of Groningen,The Netherlands}

\begin{abstract}
This study concentrates on preserving privacy in a network of agents where each agent seeks to evaluate a general polynomial function over the private values of her immediate neighbors.
We provide an algorithm for the exact evaluation of such functions while preserving privacy of the involved agents.
The solution is based on a reformulation of polynomials and adoption of two cryptographic primitives: Paillier as a Partially Homomorphic Encryption scheme and multiplicative-additive secret sharing.
The provided algorithm is fully distributed, lightweight in communication, robust to dropout of agents, and can accommodate a wide class of functions.
Moreover, system theoretic and secure multi-party conditions guaranteeing the privacy preservation of an agent's private values against a set of colluding agents are established. 
The theoretical developments are complemented by numerical investigations illustrating the accuracy of the algorithm and the resulting computational cost. 
\end{abstract}

\begin{keyword}
Privacy\sep cryptography\sep polynomials\sep networked systems \sep multiagent systems
\end{keyword}

\end{frontmatter}

\section{Introduction}

Emerging distributed systems such as smart grids, intelligent transportation, and smart buildings provide better scalability, fault tolerance, and resource sharing compared to traditional centralized systems. 
  A distributed dynamical system rely on
  peer-to-peer data exchange between individual agents. 
  The agents wish to protect their data from being revealed since the data can contain sensitive information or can be leveraged for  disrupting the system(see the case for smart metering in \cite{van2019smart}).  Therefore preserving privacy of agents in distributed dynamical systems is of crucial concern.

\par To preserve privacy in dynamical systems, differential privacy is a popular approach that was introduced to control system for private filtering through \cite{le2013differentially}, applied to average consensus \cite{mo2016privacy, nozari2017differentially}, distributed optimization \cite{nozari2016differentially}, plug-and-play control \cite{kawano2021modular} and studied for its relation to input observability \cite{kawano2020design}.  
In general, however, it introduces a trade-off between privacy level and control performances, and also possible vulnerability of data disclosure through the least significant bit of the perturbed data \cite{mironov2012significance}.
System theory also provides solutions for preserving privacy in dynamical systems, see \cite{altafini2020system, sultangazin2020symmetries, monshizadeh2019plausible} in this context. Even though these solutions do not generally degrade the performance of controllers and are lightweight in computation, they are problem specific and their privacy guarantees are weaker compared to differential privacy based methods.

\par \blue{Recently, cryptography based methods have gained attraction in control systems as a new avenue for privacy problems.}
We can classify the literature in this area into two main categories: 
 The first one is the typical setup of an isolated system and a cloud 
 where the cloud evaluates a controller 
 using the encrypted data generated by the system.
Among the first known studies of encrypted control in this group, we can refer to \cite{Kogiso2015} where the privacy of controller parameters and system states are preserved by using RSA and ElGamal encryption schemes for static state feedback controllers. \blue{The heavy involvement of the system in the computation procedure is resolved in \cite{Farokhi2017a} by employing Paillier's scheme as a Partially Homomorphic Encryption(PHE) and extended to nonlinear state feedback in \cite{SchulzeDarup2020} using a framework of two non-colluding clouds combined with Paillier's scheme.}
Deployment of linear dynamic controllers over a cloud is investigated in \cite{Kim2016} by employing Fully HE(FHE). The essence of recursion in these controllers causes the finite-time life span problem, for which the necessity of integer coefficients in \cite{Cheon:2018}, and refreshment of the controller state in \cite{Murguia2020} are proposed as possible solutions. Outsourcing the calculation of computationally demanding controllers such as  the implicit model predictive controller to another party is also investigated in \cite{Alexandru2020a} where the privacy problem is resolved mainly by Paillier's scheme.

\par The second category, which this study also belongs to, is related to
 preserving privacy in multiagent systems \blue{where compared to the first case, the communication topology and the presence of different agents (parties) impose additional constraints on the problem.}
In this group, authors in \cite{Ruan2019} use Paillier and weight decomposition to preserve privacy in the first order consensus problems where the solution is extended to the second order case in \cite{fang2021secure}. Their proposed method is suitable when an agent's objective is to evaluate an affine function of her immediate neighbors and is restricted to consensus problems. The privacy preservation is further investigated in \cite{Hadjicostis2020} for the case of directed communication among the agents.
In the context of distributed optimization, authors in \cite{Lu2018} use a symmetric FHE scheme (SingleMod encryption) and a third party to preserve privacy of the involved agents interested in evaluating polynomial functions of private variables.
The existence of a central non-colluding third party poses this question that whether it is possible to solve the problem in a centralized instead of a distributed way.
In other words, an FHE scheme and the third party allow the designer of the optimization scheme to devise a centralized algorithm.
\par  The problem of privacy in cooperative linear controllers in networked systems is first considered in \cite{Darupcooperative} and later its issue of information revealing is resolved in \cite{Alexandru2019} by combining additive secret sharing with Paillier's scheme.
The problem is viewed more generally as a private weighted sum aggregation and discussed further in \cite{alexandru2021private}.   
We refer the interested readers to \cite{Darup_Alexandru} for an overview of the recent applications of cryptography in dynamical systems.
\par In this work, we consider the problem of privacy in the networked systems where each agent's goal is to evaluate a general polynomial of her neighboring agents' private values.
Our solution is based on two cryptographic primitives: the Paillier encryption technique and secret sharing.
Paillier cryptosystem is a public key PHE \cite{paillier} which allows us to evaluate the sum of two values of plaintext using their ciphertext and secret sharing enables us to distribute shares of a secret among agents in the network. 
The main contributions of this study are as follows:\footnote{Preliminary results of this work are presented in \cite{Hosseinalizadeh}.
Different from the conference article, this document presents distributed secret sharing using pseudorandom functions, considers all polynomial coefficients as private, investigates robustness to agent dropouts, provides a formal proof for Theorem \ref{the1}, analyzes privacy from a system theoretic perspective (see section \ref{sec:pri-ana}), and provides  motivating examples as well as a new case study.}

\par (i) The current work extends the class of computed functions from affine \cite{Darupcooperative, Alexandru2019, alexandru2021private}, which is customary in linear averaging protocols, to general polynomials, i.e. each agent's target function is a polynomial function of her neighbors' state variable. \blue{As such, the proposed protocol finds its way to networked control and optimization with \textit{nonlinear} coupling law. Instances of those include, but not limited to, game theoretic controllers \cite{yi2019operator}, general consensus on nonlinear functions \cite{cortes2008distributed}, collision avoidance  \cite{mylvaganam2017differential}, and optimal frequency controllers \cite{dorfler2017gather}. The extension to the polynomial case, particularly due to the products of the state variables, substantially complicates the problem and requires a careful analysis to ensure that no privacy sensitive information is leaked throughout the computation.} \blue{We note that the proposed solution evaluates a polynomial function as a whole without revealing the values of any of its components.}
(ii) \blue{Motivated by applications in distributed control, networked control, and distributed optimization, we respect the sparsity of the communication graph in the proposed protocol.
Our algorithm is \textit{fully} distributed, and hence 
 enables the agents themselves to evaluate polynomial functions---obviating the need for external parties, e.g. those used in \cite{Lu2018}.} Furthermore, the proposed algorithm is robust to dropout of an agent and lightweight in communication due to the adopted schemes from cryptography.

(iii) We establish conditions for privacy preservation of an agent with respect to a set of colluding agents for the proposed algorithm using both cryptography and system theory paradigms. \blue{In particular, privacy analysis results reported in Subsection \ref{sec:global_collude} directly descends from a networked control system point of view where multiple functions need to be computed, one for each agent. Results such as those reported in Theorem \ref{prp:net_coll} are absent in the cryptography literature, whereas they become relevant in networked control/optimization.}

\par The rest of the paper is organized as follows: In Section \ref{sec:not-pre}, we present necessary cryptographic tools for the paper; Section \ref{sec:prb-for} \blue{includes motivating examples} for polynomials and formulates the problem of preserving privacy for these functions, and Section \ref{sec:pro-alg} provides a solution to this problem. Privacy analysis of the proposed method is investigated in Section \ref{sec:pri-ana}; numerical results are provided in Section \ref{sec:sim}, and the paper closes with conclusions in Section \ref{sec:con}.
\section{ Notations and preliminaries }\label{sec:not-pre}
The sets of positive integer, nonnegative integer, integer, rational and real numbers are denoted by $\mathbb{N}$, $\mathbb{N}_0$, $\mathbb{Z}$, $\mathbb{Q}$, and $\mathbb{R}$, respectively.
We denote the identity matrix of size $n$ by $I_n$ and we  write $[n] :=\{ 0, 1, 2, \ldots, n\}$ for any $n \in \mathbb{N}$.
We assume a network of $\blue{N}$ agents represented by an undirected graph $\mathcal{G}(\calV, \blue{\calW})$, with node set $\calV=\{1,2, \ldots ,N\}$ and edge set $\calW$ given by a set of unordered pairs $\{i, j\}$ of distinct nodes $i$ and $j$.
We denote the set of neighbors of node $i$ by ${\calN_i}:=\{j\in \calV: \{i, j\}\subseteq \calW\}$, and  $\ONI :=\calN_i \cup {\color{blue}\{}i{\color{blue}\}}$. The cardinality of $\calN_i$ denoted by $d_i :=\left| {\calN_i} \right|$ is equal to the degree of node $i$.
Without loss of generality we consider the state variable of each agent $i$ as a scalar $x_i \in \mathbb{R}$; the extension to $x_i \in \R^{n_i}$, $n_i \geq 2$ is straightforward. We collect the state variables of all agents as $x := \col(x_1, x_2, \ldots, x_N) = [x_1^{\top}, \ldots, x_N^{\top}]^\top$ and the state variables of all agents except for agent $i$ as $x_{-i} := \col(x_1, \ldots, x_{i-1}, x_{i+1}, \ldots, x_N)$. 
\subsection{Cryptography primitives}\label{Cryptography}
 The Paillier encryption scheme consists of three steps ({\Gen}, {\Enc},  {\Dec}). 1) {\Gen}: Given the bit-length ($l$), generates {\color{blue}$(\scrN, \scrP, \scrQ)$} where $\scrN=\scrP\scrQ$ and $\scrP$ and $\scrQ$ are randomly selected $l$-bit primes, 2) {\Enc}: Given public key $\text{pk} = \scrN$ and a message $m\in {\Z_{\scrN}}$, pick uniform $r\leftarrow \mathbb{Z}_{\scrN}^{*}$ and output $c :=[{{(1+\scrN)}^{m}}.{{r}^{\scrN}}\,\bmod \,\scrN^2]$ as ciphertext, 3) {\Dec}: Given the secret key $\text{sk}=\phi (\scrN)=(\scrP-1)(\scrQ-1)$ and the ciphertext $c$ computes $m := \left[ \frac{[{{c}^{\phi (\scrN)}}\bmod \,\scrN^2]-1}{\scrN}.\phi {{(\scrN)}^{-1}}\,\bmod \scrN \right]$. Paillier encryption scheme is chosen plaintext attack secure based on hardness of decisional composite residuosity problem  \cite[p.~495-496]{katz}. 
It is easy to see that for any plaintext $m_1$ and $m_2$ and their respective encryptions $c_1$ and $c_2$, we have {\Dec}$({c_1} \cdot {c_2})= {m_1} + {m_2}$, i.e. the Paillier scheme is an additively HE also known as a PHE. We denote the encryption of a value ${m}$ by agent $i$'s public key ${\mathrm{pk}_i}$ as $\ENC {m}_{\mathrm{pk}_i}$ or {\color{blue}simply $\ENC {m}$ when the key is clear from the context}.
\par In $(t, n)$-threshold secret sharing, \blue{the aim is} to share a secret $s$ among some set of $n$ agents $a_1$, $a_2$, $\ldots$, $a_n$ by giving each one a share in such a way that only $t$ or more users can reconstruct the secret. In other words, no coalition of fewer than $t$ agents should get any information about $s$ from their collective shares. When $t = n$ and $s$ is $l$ bit length $s\in {{\{0,1\}}^{l}}$, \blue{the shares $s_1,\ldots,s_{n-1}\in {\{0,1\}}^l$ are chosen uniformly randomly while $s_n = s\oplus {\big( \oplus _{i=1}^{n-1}s_i \big)}$}, where $\oplus$ denotes bitwise exclusive \cite[p.~501-502]{katz}. The share of agent $a_i$ is $s_i$. 

\par Paillier scheme only accepts nonnegative integers $\mathbb{N}_0$ in its domain while we are interested in computations over $\mathbb{R}$. Therefore, an encoding-decoding scheme satisfying homogeneity and additivity conditions is incorporated. This scheme first 
discretizes a value ${v} \in \mathbb{R}$
to $\hat{{v}}\in \mathbb{Q}$, then maps it to an integer ${z}\in \mathbb{Z}$ by choosing an appropriate scale $L\in \mathbb{N}$, then to a nonnegative integer using ${m}={z}\,\,\bmod \,\Omega $, where $\Omega$ is a sufficiently large number. This process is invertible, i.e. given  $0 \le {m} \le \Omega$ we can recover  $\hat{{v}}\in \mathbb{Q}$.
To simplify the notation, we do not distinguish between the encoded value and the original real value of a quantity in modular operations.
\section{Motivating examples and problem formulation} \label{sec:prb-for}
We begin with examples motivating privacy considerations in computation of polynomials in networks. The problem formulation will be stated afterwards.
\subsection{Motivating examples}\label{Motivation}
Control of networked systems and optimization on networks heavily rely on computation of functions 
over state variables of neighboring agents. As mentioned earlier, the focus of this work is on private computation of polynomial functions. 
As a case in point, we provide an example in the context of game theoretic algorithms {\color{blue} followed by the problem of consensus on general functions.}
 
\subsubsection{Game theoretic algorithms}\label{subsec:game}

Consider a group of $N = |\calV|$ players (agents) that 
seek a generalized Nash equilibrium (GNE) of a noncooperative game with globally shared affine constraints \cite{yi2019operator}. Feasible decision set of all players is  
$X :=\prod_{i=1}^{N} \Gamma_i \cap \big \{x \in \R^{n}:\sum_{i=1}^{N} A_i x_i \geq \sum_{i=1}^{N} b_i\big\}$ where $ x_i$ takes values from a local admissible set $\Gamma_i \subseteq \mathbb{R}^{n_i}$, and $A_i \in \R^{m \times n_i}$, $b_i \in \R^{m}$ are local parameters of player $i$.
In this game, each player $i$ aims at minimizing her local cost function $J_i(x_i, x_{-i})$ subject to her feasible decision set $X_i(x_{-i}) := \big\{ x_i \in \Gamma_i: (x_i, x_{-i}) \in X\big\}$, i.e.,
\begin{equation}\label{e:game}
    \min_{x_i \in \R^{n_i}}J_{i}(x_i,x_{-i}) \quad \text{s.t.}\,\, x_i \in X(x_{-i}).
\end{equation}
A distributed GNE seeking algorithm is proposed in \cite{yi2019operator} where at step $1$ of this algorithm, each player $i \in \calV$ updates her decision at time index $k$ as
\be\label{e:grad}
    x_i(k+1) = \text{proj}_{\Gamma_i}\big[x_i(k) - \tau_{i}\big(\nabla_{x_i}{J_i(x_i(k), x_{-i}(k))}- A_i^\top\lambda_i(k)\big)\big],
\ee
where $\text{proj}_{\Gamma_i}(\cdot)$ is the Euclidean projection operator onto the set $\Gamma_i$, $\lambda_i$ is the Lagrange multiplier, and $\tau_i$ is the step size of player $i$. 

\par The gradient of the cost function $\nabla_{x_i}{J_i(x_i, x_{-i})}$  is generally a nonlinear function of decision variables of other players $x_{-i}$.
Therefore, player $i$ in the game generally needs the value of $x_{j}$ with $j \in \calV$ for running \eqref{e:grad}. Putting it differently, 
player $j$ must share her decision variable $x_j$ with player $i$. Sharing the decision variable $x_{j}$ over time can reveal information on the cost function $J_{j}(x_j,x_{-j})$ which includes privacy sensitive parameters of player $j$.

\par In the case of quadratic cost functions, $\nabla_{x_i}{J_i(x_i, x_{-i})}$ is affine and hence the scheme developed in \cite{alexandru2021private} can be used to evaluate it privately. Moreover, in the context of aggregative games,  \cite{gade2020privatizing} and \cite{shakarami2019privacy} 
have proposed noncryptographic based solutions in order to preserve privacy of decision variables. 
The quadratic costs and aggregative functions in the mentioned studies are special cases of the polynomial functions that we consider here.
  
{\color{blue}\subsubsection{Consensus on general functions}\label{subsec:gen_con}
Consensus is one of the fundamental elements in control of network systems with applications in cooperative control, network games, data fusion and distributed filtering. 
In consensus on general functions \cite{cortes2008distributed},
a continuous function $J: \R^N \to \R$ is considered for a set of $N = |\calV|$ agents connected through a (directed) communication graph $\mathcal{G} = (\calV, \calW)$
each with the dynamic
\[
\dot{x}_i = u_i, \qquad i\in\{1,\ldots,N\},
\]
where $x_i \in \R$. 
The problem is stated as designing $u_i: \R^N \to \R$ such that the agents reach consensus on $J(x_1(0), \ldots, x_N(0))$, meaning $x_i(t) \to J(x_1(0), \ldots, x_N(0))$ as $t \to \infty$.
It is shown in \cite[Prop. 10]{cortes2008distributed} that if each agent $i \in \calV$ applies the control input
\be\label{eq:gen_con}
u_i = \frac{1}{|\nabla_{x_i}{J}|}\sum_{j=1}^{N}{a_{ij}(x_i - x_j)},
\ee
the agents reach consensus on the general function $J$.
\par Evaluating $u_i$ requires that each agent $i \in \calV$ computes a function of other agents' state quantities which are privacy sensitive.  
 Depending on the form of the general function $J$, the control input $u_i$ can be a nonlinear function due to the term $|\nabla_{x_i}{J}|$, thereby motivating the need to develop and equip the agents with a protocol that can privately evaluate nonlinear functions. It is worth mentioning that conventional consensus becomes a special case of \eqref{eq:gen_con}, where the function $J$ becomes the arithmetic mean of initial conditions.

 Other prominent applications of nonlinear neighboring functions in networked systems are differential game approaches to multiagent collision avoidance in \cite{mylvaganam2017differential} where the agents are required to share their physical coordinates, as well as game-theoretic controllers in physical networks  \cite{dorfler2017gather} in which agents need to compute and communicate their nonlinear marginal cost functions.}
\par In this work, we propose a cryptography-based algorithm that enables a private computation of \blue{\textit{general}} polynomials over networks, thereby preserve privacy for a broad range of nonlinear functions appearing in game-theoretic\blue{, as well as networked control and optimization problems.} 
 The choice of studying polynomial functions is further motivated by the fact that
 any function continuous on a closed bounded set can be approximated by a polynomial with a desired accuracy (see the Stone-Weierstrass approximation theorem, e.g. \cite[p.~123]{dzyadyk2008theory}.)
\subsection{Problem formulation}\label{subsec:prb-for}
We consider a scenario where at each time index $k \in [K]$  agent $i \in \calV$ in the network $\mathcal{G}$ is interested in evaluating a $d \in \mathbb{N} $ degree polynomial $\calP_{i}(x_i, x_{\calN_i}): \mathbb{R} \times \mathbb{R}^{|\calN_i|} \to \mathbb{R}$; namely, 
\be\label{e:poly_gen}
\begin{aligned}
 \calP_{i}\big(x_i(k), x_{\calN_i}(k)\big):= 
 \sum\limits_{(p_1,p_2,\ldots,p_m) \in \calX_i} c_{(p_1p_2\ldots p_m)}x_1^{p_1}(k)x_2^{p_2}(k)\ldots x_m^{p_m}(k)
\end{aligned}
\ee
where $c_{(\cdot)} \in \R$, {\color{blue}$x_{\calN_i} := \col(x_j)_{j \in \calN_i}$ is the state variables of agent $i$'s neighbors} and 
\begin{align*}
   &\calX_i := \{(p_1,p_2,\ldots,p_m) \in \mathbb{N}_0^m: \\
    & p_1 + p_2 + \ldots + p_m \le d, j \notin \ONI \Rightarrow p_j = 0\}.
\end{align*}
As we can see \eqref{e:poly_gen} depends not only on agent $i$'s state variable $x_i$ but also on the state variable of her neighbors $x_j$ with $j\in \mathcal{N}_i$. 
We identify private and public values of the agents as: Private value of agent $i$ is $\PV_i := \{ x_i, c_{(\cdot)}\}$ that includes her state variable $x_i$ and all the coefficients $c_{(\cdot)}$ in \eqref{e:poly_gen}, private value of agent $j$ is $\PV_j := \{ x_j\}$ with $j \in \calN_i$, and public values of agent $i$ is the exponent of state variables in \eqref{e:poly_gen}, i.e., $\calX_i$. Notice that the agent $i$ shares the public values with agents $\calN_i$. 
\par We now provide two privacy assumptions that clarify the adopted setup in our problem formulation.
\begin{assumption}[Honest-but-curious]\label{as:hon}
Agents in a network $\mathcal{G}$ are honest-but-curious, also known as semi-honest, meaning that they follow the required protocol for interacting with other agents but are also interested in determining the private values in the network.
\end{assumption}
\begin{assumption}[Passive Adversary]
 An adversary $\mathcal{A}$ is probabilistic polynomial-time, passive, and communications among agents are done in her presence. The adversary $\mathcal{A}$ can be an agent in the network or an external party observing the communication.
\end{assumption}

\blue{\textit{Research goal.}} Our aim is to provide a \textit{privacy} preserving protocol for the \textit{exact} evaluation of (\ref{e:poly_gen}) for agent $i$. That is to say, only agent $i$ should be able to obtain the accurate value of  $\calP_{i}$
without revealing her own private value $\PV_i$ to any other agent $j$ and without gaining any privacy-sensitive information about $\PV_j$ other than the target function $\calP_{i}$.

\section{Proposed algorithm}\label{sec:pro-alg}
The solution we provide is based on PHE and secret sharing techniques.
In particular, we use Paillier's scheme to protect the privacy of $\PV_i$, and secret sharing for preserving the privacy of $\PV_j$, with $j \in \mathcal{N}_i$. For adopting these schemes in our privacy preserving algorithm, we rewrite the polynomial  \eqref{e:poly_gen} in a new from by making a distinction between its bivariate and multivariate terms. %
Namely, we write \eqref{e:poly_gen} as
\be\label{e:poly_extend}
\begin{aligned}
  \calP_{i}(x_i, x_{\calN_i})= \sum\limits_{j\in \calN_{i}}P_{j}(x_i, x_j)+ 
  Q_i(x_i, x_{\calN_i})
\end{aligned}
\ee
where
\begin{align*}
&P_{j}(x_i, x_j)=\sum\limits_{\p, \pj}{c_{p_ip_j}x_i^{\p}x_j^{\pj}}, 
\end{align*}
with $c_{p_ip_j} \in \R$, and $Q_i(\cdot, \cdot)$ contains the terms with at least two state variables from $x_{\calN_i}$.
Notice that $P_{j}$ is the summation of bivariate terms in \eqref{e:poly_gen}, with $x_i$ and $x_j$, $j\in \mathcal{N}_i$, being the corresponding two variables. {\color{blue} We dropped the argument $k$ in $x_i(k)$ and $x_j(k)$ to simplify the notations.}

As will be observed, our algorithm leverages additive and multiplicative secret sharing to preserve the privacy of neighbors of $i$ in the evaluation of \eqref{e:poly_extend}. The additive secrets will be primarily used to evaluate the bivariate terms 
in \eqref{e:poly_extend} whereas the multiplicative secret sharing is exploited to evaluate the multivariate terms, i.e. $Q_i(\cdot, \cdot)$. Motivated by this and to minimize the required communication in the protocol, we write $Q_i$ as a summation of multiplicative terms, namely:
\be\label{e:pol_mult}
 Q_i = \sum_{t=1}^{T}{Q_i^t} , \qquad Q_i^t(x_i, x_{\calN_i}) :=\prod\limits_{j\in \ONI}W_j^t(x_j)
\ee
with $W_j^t=\sum_{\qj}{c_{\qj}^{(t)}x_j^{\qj}}$, $Q_i^t (\cdot, \cdot)\ne 0$, $T \in \mathbb{N}$, and $c_{\qj}^{(t)} \in \R$. Note that $W_j^t$ is a univariate polynomial of $x_j$. {\color{blue}This will be further illustrated in Example \ref{Ex:ex1}.}

Private and public values in  \eqref{e:poly_extend} remain the same as in \eqref{e:poly_gen}; yet take the form   $\PV_i^{} := \{ x_i, c_{{\p}{\pj}}, c_{\qj}^{(t)}\}$  and $(p_j, q_j)$ for private and public values of agent $i$, respectively, and $\PV_j := \{ x_j\}$ with $j \in \calN_i$.
\subsection{Distributed secret sharing}\label{subsec:dis_sec}
As mentioned before, we use secret sharing for preserving the privacy of $x_j$, $j\in \mathcal{N}_i$, throughout computation of $\calP_i$. In particular, additive secret sharing is used in the bivariate part (namely, $P_{j}$) and multiplicative secret sharing over additively secret shared data is used in the multivariate terms in $\calP_i$ (namely, $Q_i$). 
 
\par To mask intermediate computations \cite{bonawitz2017practical}, we require that every agent $j \in {\ONI}$ to have shares of multiplication ${\color{blue}\sm_j}(k)$ and addition ${\color{blue}\sa_j}(k)$ for all $i \in \calV$ and for all $k \in [K]$  such that
    \bse\label{e:sec}
\be\label{e:sec_add}
\prod\limits_{j\in {\ONI}}{\sm_j(k)}\equiv 1\,\,\,\bmod \Omega
\ee
\be\label{e:sec_mul}
\sum\limits_{j\in {\ONI}}{\sa_j(k)}\equiv 0\,\,\,\bmod \Omega,
\ee
\ese
 where $\Omega$ is a publicly known and sufficiently large prime number and $\sm_j\ne 0$. 
The shares $\sm_j$ and $\sa_j$  are selected uniformly randomly $\forall k \in [K]$ from the following set:
\be\label{Gal}
{\Z_{\Omega }}=\left\{ 0,1,\ldots,\Omega -1 \right\}.
\ee
 It should be noted that based on the Fermat's little theorem every nonzero element in (\ref{Gal}) has a multiplicative inverse, meaning $(\forall \omega \ne 0 \in \mathbb{Z}_{\Omega})(\exists \omega^{-1}\in \mathbb{Z}_{\Omega})$ such that $(\omega \omega^{-1} \equiv 1 \quad \bmod \Omega)$, therefore choosing $\sm_j$ and $\sa_j$ in the required form \eqref{e:sec} is feasible \cite[p. 63]{hardy1979introduction}.

\par We note that although the ``secrets" ($0$ and $1$) are known to the agents,
the generation of the shares $\sm_j$ and $\sa_j$ is analogous to sharing of a secret $s$ explained in Subsection \ref{Cryptography}, and as such, we occasionally refer to this scheme as secret sharing (see also \cite{alexandru2021private} where a similar terminology is used for additive secret sharing). {\color{blue}In what follows, all computations are in $\bmod \, \, \Omega$, unless specified otherwise.}

\par Next, we generate the multiplicative and additive shares of the agents in a fully distributed manner. 
To this end, every agent $j \in \ONI$ selects uniformly randomly $\sm_{jh}$ and $\sa_{jh}$ from \eqref{Gal} $\forall k \in [K]$ such that
    \bse\label{e:sec_dis_1}
\be\label{e:add_dis_1}
\prod\limits_{h\in {\ONI}}{\sm_{jh}}\equiv 1
\ee
\be\label{e:mu_dis_1}
\sum\limits_{h\in {\ONI}}{\sa_{jh}}\equiv 0,
\ee
\ese
where $i \in \calV$. Then agent $j$ sends $\sm_{jh}$ and $\sa_{jh}$ for ${h\in \ONI\backslash j}$ to agent $i$ through a secure communication channel $\forall k \in [K]$, where agent $i$ then sends each share to its corresponding receiver, agent $h$. After this step, agent $j$ obtains the multiplication and addition $\forall k \in [K]$ shares as follows:
    \bse\label{e:sec_dis_2}
\be\label{e:mu_dis_2}
\sm_{j} := \prod\limits_{h\in {\ONI}}{\sm_{hj}} 
\ee
\be\label{e:add_dis_2}
\sa_{j} := \sum\limits_{h\in {\ONI}}{\sa_{hj}}.
\ee
\ese
Notice that the distributed shares $\sm_j$ and $\sa_j$ obtained in \eqref{e:mu_dis_2} and \eqref{e:add_dis_2} satisfy the relations \eqref{e:sec_add} and \eqref{e:sec_mul}, respectively, as desired.
\subsection{Distributed secret sharing using pseudorandom functions}\label{subsec:sec_psuedo}
Exchanging $|\calN_i| \times |\calN_i|$ random numbers in \eqref{e:add_dis_1} and \eqref{e:mu_dis_1} for every time index $k$ among the agents imposes extra communication loads in the network.  
This drawback can be circumvented using the idea of Pseudorandom Functions(PRFs)\cite[p.~79-159]{ lindell2017tutorials}.
A pseudorandom function $F: \{0,1\}^{l} \times \{0,1\}^{l} \to \{0,1\}^{l}$, where $\{0,1\}^{l}$ denotes an $l$-bit sequence, accepts two arguments as its inputs: a key $\kappa$ and a seed $\gamma$ and returns a random number $\rho$. PRFs cannot be differentiated from truly random functions by any efficient procedure that can get the values of the functions at arguments of its choice.
\par In order to generate multiplication $\sm_j$ and addition $\sa_j$ shares for $j \in {\ONI}$ and $\forall k \in [K]$ using PRFs, every agent $j \in {\ONI}$ randomly selects \blue{$\Fm_{jh}$} and \blue{$\Fa_{jh}$} from \eqref{Gal} such that
    \bse\label{e:PRF1}
\be\label{e:PRF_add_1}
\prod\limits_{h\in {\ONI}}{F(\Fm_{jh}, \gamma_k)}\equiv 1
\ee
\be\label{e:PRF_mul_1}
\sum\limits_{h\in {\ONI}}{F(\Fa_{jh}, \gamma_k)}\equiv 0,
\ee
\ese
for some $\gamma_k\in \mathbb{Z}_{\Omega}$ to be specified later. Then, agent $j$ sends $\Fm_{jh}$ and $\Fa_{jh}$ to ${h\in \ONI\backslash j}$ through agent $i$.

After this step, agent ${j\in \ONI}$ computes
    \bse\label{e:PRF2}
\be\label{e:PRF_mul_2}
F(\Fm_{j}) := \prod\limits_{h\in {\ONI}}{F(\Fm_{hj}, \gamma_k)} 
\ee
\be\label{e:PRF_add_2}
F(\Fa_{j}) := \sum\limits_{h\in {\ONI}}{F(\Fa_{hj}, \gamma_k)} 
\ee
\ese
to obtain the PRFs that are needed for generating her random shares. Agent $j \in \ONI$ then is able to get $\sa_{j}(k)$ and $\sm_{j}(k)$ by evaluating $F(\Fa_{j}, \gamma_k)$ and $F(\Fm_{j}, \gamma_k)$ for a specific seed $\gamma_k$.
\par We remark that  \eqref{e:PRF1} and \eqref{e:PRF2} \textit{only need to be executed once}, and the agents
do not need to communicate with each other after receiving the key $\kappa$. 
Moreover, we note that the (initial) seed is a public value, is the same for all agents ${\ONI}$, and should be distinct for every time index $k\in [K]$. To ensure this, before the start of the protocol, the agents can agree on a public value, namely $\gamma_0 = S \in \mathbb{Z}_{\Omega}$,
and then use 
$\gamma_k = F(., s(k))$ as the seed for all time (\blue{see Remark \ref{rem:com_shares} for further discussion)}.
\begin{example}\label{Ex:ex1}
As an example assume that we have a network with $\{\{1,2\}$,$\{1,3\}$,$\{1,4\}\}\subseteq \calW$ where agent 1 is interested in the evaluation of the following polynomial:
\begin{equation}\label{e:ex1}
\begin{aligned}
    {\calP}_{1}(x_1, x_2, x_3, x_4)=
  2x_1^2x_2 + 3x_1x_3 + 4x_1x_4^3 + x_1x_2^2x_3^2x_4 + 3x_1x_2^2x_3x_4. 
\end{aligned}
\end{equation}
Based on the representation \eqref{e:poly_extend}, we can specify bivariate parts as $P_{2} = 2x_1^2x_2$, $P_{3} = 3x_1x_3$, $P_{4} = 4x_1x_4^3$.
To write the multivariate parts, after factoring out the term $x_1x_2^2x_4$, we obtain $W_1^1 = x_1$, $W_2^1 = x_2^2$, $W_3^1 = x_3^2 + 3x_3$, and $W_4^1 = x_4$.
{\color{blue}Agents $1, 2, 3$, and $4$ generate  multiplicative and additive shares $\sm_j$, $\sa_j$ for $j \in \{1, 2, 3, 4\}$ either through \eqref{e:sec_dis_2} or \eqref{e:PRF2}.
This is illustrated in Fig. \ref{f:dis_sec}, where the generated shares, via \eqref{e:sec_dis_2}, are depicted for agent $2$ as an example.
She generates the required shares through \eqref{e:sec_dis_2} and sends them to agent $1$'s neighbors. Note that, as mentioned before, the communication between agent $2$ and each of agent $1$'s neighbors is carried through a secure channel.} 
 \begin{figure}[!h]
  {\color{blue}
\begin{center}
\includegraphics[width=0.35\textwidth]{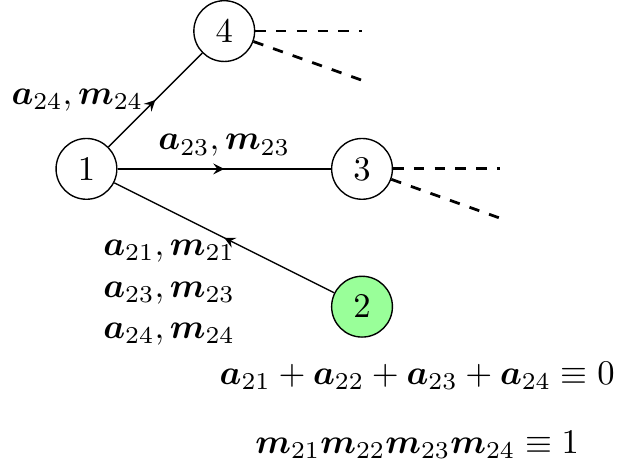}
\caption{Generation of multiplication and additive shares by agent $2$}\label{f:dis_sec}
\end{center}}
\end{figure}
\EXND
\end{example}
{\color{blue} \begin{remark}[Complexity of generating the secret shares]\label{rem:com_shares}
The required random shares $\sa$ and $\sm$ can be generated either through \eqref{e:sec_dis_2} or \eqref{e:PRF2}. Each of these choices may be preferred given the available computational resources and the communication constraints.
In terms of communication complexity, generating $\sa$ and $\sm$ by \eqref{e:sec_dis_2} in essence requires each agent to send a random number($l$-bit) to each neighbor of $i$ for every time index $k$ and thus communication complexity becomes $\calO(|\calN_i|^2lK)$. On the other hand, for generating $\sa$ and $\sm$ by \eqref{e:PRF2} each agent sends the key $\kappa$($l$-bit) of a pseudorandom function for the total time index $K$; therefore, the communication complexity becomes $\calO(|\calN_i|^2l)$. Notice, the reduction of $K$ in the second case, which shows the benefit of pseudorandom function scheme in terms of communication complexity. 
As for computational complexity, block cipher AES (Advanced Encryption Standard) which is a symmetric key encryption scheme is used in practice as the pseudorandom function $F(\cdot, \cdot)$ in \eqref{e:PRF2}. The details for AES can be found in \cite{daemen2002design} 
    which roughly consists of key addition, byte substitution, and diffusion layers. Computational complexity of generating random shares through \eqref{e:PRF2} is proportional to the cost of evaluating the function $F$. On the contrary, generating $\sa$ and $\sm$ by \eqref{e:sec_dis_2} is done by selecting the shares $\sa$ and $\sm$ uniformly randomly from the set \eqref{Gal}, which is much lighter computationally. 
\end{remark}}
\subsection{The protocol}\label{subsec:protocol}
Now that the additive and multiplicative shares are generated, we provide an algorithm that enables the private computation of $\calP_i$ in \eqref{e:poly_extend}.
To simplify the presentation and ease the notation, we discuss the required steps for the case that $T=1$ (see \eqref{e:pol_mult}) \blue{and we drop $t$ in the sequel.}
{We explain in Remark \ref{rem:extension} how the proposed algorithm can be extended to the case $T >1$.}

The formal steps of the algorithm is provided in the next page (see Algorithm \ref{Alg} \blue{and Fig. \ref{f:pro_schem}}).
To differentiate between a generic variable $x_i$ 
($x_j$, respectively) and its particular value at a given time index, we denote the latter by $\bld{x_i}$ ($\bld{x_j}$).
Recall that the structure of the involved polynomial functions of $(x_i, x_j)$ are known but the values $\bld{x_i}$ ($\bld{x_j}$) are considered private. 
Moreover, among the neighbors of agent $i$ a specific agent denoted by ${{D}_{i}}\in {\calN_i}$ is distinguished and her role becomes clear later (see Step 5 and Remark \ref{rem:3}).
\begin{enumerate}[S1), leftmargin=*, labelindent=0pt]
\item At the start of the algorithm, agent $i$ chooses independently her public key ${\mathrm{pk}_i}$ and private key $\ski$ for the Paillier scheme; then publishes her public key ${\mathrm{pk}_i}$. 
\item Agent $i$ uses her public key ${\mathrm{pk}_i}$ to encrypt her private quantities that appear in $P_{j}(x_i, x_j)$ and $W_j(x_j)$, and sends the corresponding encrypted terms, namely $ \ENC{c_{p_ip_j}\bld{x}_{\bld{i}}^{\p}} $ to agent $j\in {\calN_i}$ and $ \ENC{c_{\qj}} $ to agent $j\in{(\calN_i \setminus D_i)}$ for all ${\p}$, ${\pj}$ and $\qj$ where ${\p}$, ${\pj}$ and $\qj$ are the exponents of the respective polynomial for the corresponding agent. The reason behind this encryption is elaborated in Remark \ref{r:enc-why}.
Moreover, she computes ${{\mu }_{i}}=\,\,\left( \sm_i{{W}_i}(\bld{x_i}) \right)$  and records it for Step $4$.
\item Every agent $j\in \left( {\calN_i}\backslash {{D}_{i}} \right)$ encrypts $\sa_j$ using ${\mathrm{pk}_i}$ and evaluates the following expressions over the ciphertext
\bse
\be\label{e:not_last1}
\sigma_j=\,{{\prod\limits_{\p, {\pj}}{ 
 \ENC{c_{p_ip_j}\bld{x}_{\bld{i}}^{\p}}^{\bld{x}_{\bld{j}}^{\pj}}}}}{ \ENC{\sa_j} }\bmod \,\,\scrN^2
\ee
\be\label{e:not_last2}
 \ENC{\mu_j} =\,{\prod\limits_{{\qj}}{ 
 \ENC{c_{\qj}}^{\sm_j\bld{x}_{\bld{j}}^{\qj}}}} \bmod \,\, \scrN^2,
\ee
\ese
then sends $\sigma_{j}$ and $ \ENC{\mu_j} $ to agent $i$.
By the end of this step, all computations from the agents $j\in (\calN_i\setminus D_i)$ are carried out\footnote{In case coefficients $c_{\qj}\text{'s}$ are not privacy sensitive, then agent $j$ computes $\mu_j=\big(\sum\limits_{\qj}{c_{\qj}x_j^{\qj}}\big)\sm_j$ in \eqref{e:not_last2}.}. 
\item Agent $i$ decrypts $ \ENC{\mu_j} $ received in Step $3$ using $\ski$, computes the value
\be\label{e:mul}
{{\Psi }_{i}}=\prod\limits_{\begin{smallmatrix} 
 j\in {({\ONI})\setminus D_i}
\end{smallmatrix}}{{{\mu }_{j}}},
\ee
and sends the encrypted values $ \ENC{c_{\qj}\Psi_{i}}  $ with $j= D_i$ and for all  $\qj$ to agent $j= D_i$
\item Agent $j = D_i$,  using the values received in Step $4$, computes
\be\label{e:last_node1}
\Psi_{j} =\prod_{\qj} \ENC{c_{\qj}\Psi_{i}}  ^{(\sm_j{\bld{x_j}^{\qj})}}\bmod \,\,\scrN^2
\ee
and (\ref{e:not_last1}), and then sends ${{\sigma }_{j}}\Psi_{j} \,\bmod \,\,\scrN^2$
to agent $i$. The reason behind this step will be made clear in Remark \ref{rem:3}.
 \item Agent $i$ decrypts the received values in \eqref{e:not_last1} and values in Step $5$ using her secret key $\ski$ to obtain

\begin{equation*}
   \begin{aligned}
 &P_j(\bld{x_i},\bld{x_j})+{\sa_j} \,\,\,\,\,&{\forall}\,\,j\in \left( {\calN_i}\backslash {{D}_{i}} \right)\\
&P_j(\bld{x_i},\bld{x_j})+{\sa_j}+\prod\limits_{j\in {\ONI}}{\sm_j{{W}_j}(\bld{x_j})}\,\,\,\,\,\,&j={{D}_{i}}.  
\end{aligned}
\end{equation*}
 \item Agent $i$ sums the received results in Step $6$ and includes her own share of addition $\sa_i$ to obtain 
\[
\sum\limits_{j\in {\calN_i}}{{P_{j}}(\bld{x_i},\bld{x_j})}+\prod\limits_{j\in {\ONI}}{{{W}_j}(\bld{x_j})},
\]
where we have used (\ref{e:sec_add}) and \eqref{e:sec_mul}. After decoding, the above expression reduces to $\calP_i(\cdot)$ in \eqref{e:poly_extend} as desired.
\end{enumerate}
\begin{algorithm}[ht]
\DontPrintSemicolon
  
  \KwInput{$ \big\{\{c_{p_ip_j}, c_{\qj} \}_{j \in \calN_i}, \{x_j, \sa_j, \sm_j\}_{j \in {\ONI}}\big\}$}
  \KwOutput {Evaluation of ${{\calP}_{i}}(\bld{x_i}, \bld{x_{\calN_i}})$ given in \eqref{e:poly_extend}}
  Agent $i$ generates ${\mathrm{pk}_i}$ and $\ski$ and sends ${\mathrm{pk}_i}$ to agent $j \in \calN_i$\\
  \For{$j\in {\calN_i}$}   
  {
   \nonl agent $i$ using ${\mathrm{pk}_i}$ sends  
   $ \ENC{c_{p_ip_j}\bld{x}_{\bld{i}}^{\p}} $ to each agent $j \in \calN_i$ and sends $ \ENC{c_{\qj}} $
   to each $j \in \calN_i\setminus D_i$
   }
   \For{$j\in \left( {\calN_i}\backslash {{D}_{i}} \right)$}{
   \nonl agent $j$ computes $\sigma_j$ and $ \ENC{\mu_j} $ given in (\ref{e:not_last1}) and (\ref{e:not_last2})
   and sends the result to agent $i$ 
        }
     Agent $i$ computes $\Psi_i$ given in (\ref{e:mul}), and sends $ \ENC{c_{\qj}\Psi_{i}}  $ to agent $j = D_i$ for all $\qj$\\
     Agent $j = D_i$ computes $\sigma_j$ and $\Psi_j$ given in (\ref{e:not_last1}) and (\ref{e:last_node1}), then sends ${{\sigma }_{j}}\Psi_j\,\bmod \,\,\scrN^2$ to agent $i$ \\
Agent $i$ decrypts the received messages from her neighbors using $\ski$\\
Agent $i$ aggregates the results to obtain ${{\calP}_{i}}(\bld{x_i},\bld{x_{\calN_i}})$\\
Agent $i$ decodes the results to obtain ${{\calP}_{i}}(\bld{x_i},\bld{x_{\calN_i}})$
\caption{The protocol for private evaluation of polynomial \eqref{e:poly_extend} at time index $k$ with $T = 1$}\label{Alg}
\end{algorithm}
 \begin{figure}[!h]
 {\color{blue}
\begin{center}  
   \includegraphics[width=0.50\textwidth]{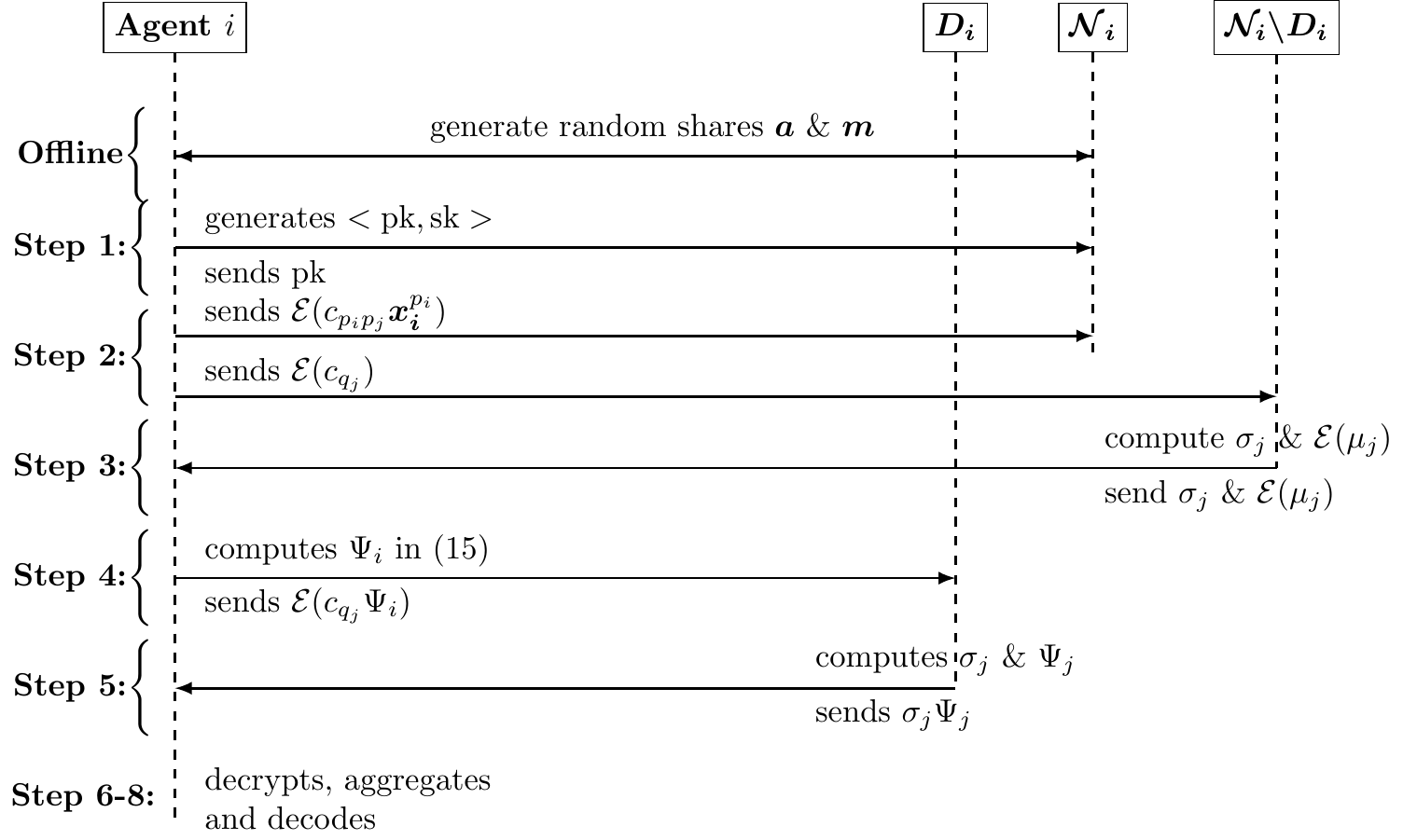}
\caption{The schematic of the protocol for private evaluation of polynomial \eqref{e:poly_extend} at time index $k$. The schematic should be read as, e.g., agent $i$ generates $<\text{pk}, \text{sk}>$ and sends pk to all her neighbors at step $1$. For other steps, similar description follows.  }\label{f:pro_schem} 
\end{center}}
\end{figure}
 A few remarks are in order concerning the proposed algorithm:
\begin{remark}[Extension to $T>1$] \label{rem:extension}
The proposed Algorithm \ref{Alg} can be easily extended to the case $T > 1$. 
This requires the agents $D_i$ and $i$ to repeat their tasks (Steps $4$ and $5$ of the algorithm) for every multivariate term $Q_i^t = \prod_{j\in {\ONI}}W_j^t(x_j)$; see \eqref{e:pol_mult}.
In this case, every agent $j \in \ONI$ also needs $T$ multiplicative shares $\sm_j$ which can be generated by \eqref{e:mu_dis_2} or \eqref{e:PRF_mul_2}.
It is worth mentioning that when $T = 0$, i.e. $\calP_i(\cdot)$ has no multivariate term, the proposed algorithm needs neither multiplicative shares $\sm_j$ nor the presence of a distinguished neighbor. Working with \eqref{e:poly_extend}, rather than \eqref{e:poly_gen}, allows us to capture this special case properly.
\end{remark}
\begin{remark}[The role of distinguished neighbor]\label{rem:3}
We designed Steps 4 and 5 of the algorithm such that the value of $\sum_{j\in \calN_{i}}P_{j}(x_i, x_j)$ and that of 
$Q_i = \sum_{t=1}^{T}{Q_i^t(x_i, x_{\calN_i})} $
in \eqref{e:poly_extend} remain hidden from both agent $i$ and $D_i$. In fact, agent $i$ can only evaluate the summation of these two terms, which amounts to the interested query in \eqref{e:poly_extend}.
Putting it differently, we can remove the distinguished neighbor from the algorithm at the expense of revealing the values $\sum\limits P_{j}(x_i, x_j)$ and $Q_i^t$ individually. 
This may not readily lead to a privacy breach for other agents, but it provides agent $i$ with extra information (beyond the query itself) that can compromise the privacy of her neighbors. Therefore, the distinguished neighbor $D_i$ should be chosen with the consensus of all neighbors of agent $i$, and without involvement of agent $i$ in this decision.
We again emphasize that the current algorithm is devised such that no information other than the query $\calP_i(\cdot)$ will be made available to the agent $i$.  
\end{remark}
\begin{remark}[Encryption]\label{r:enc-why}
It should be noted again that coefficients $c_{p_ip_j}$, $c_{{\qj}}$ and the variables $x_i$ in $P_{j}=\sum_{{\p} ,{\pj}}{c_{p_ip_j}}x_i^{\p}x_j^{\pj}$ and $W_j=\sum_{{\qj}}{c_{\qj}x_j^{\qj}}$ are sensitive data and their encryption are justified. For this reason, agent $i$ in Step $2$ of the proposed algorithm sends encrypted quantities $ \ENC{c_{p_ip_j}\bld{x}_{\bld{i}}^{\p}} $ and $ \ENC{c_{\qj}} $ to her neighbors. 
If $P_{j}$ has $\mathfrak{n}$ terms involving $x_j$ then agent $i$ has to encrypt $\mathfrak{n}$ values 
and sends them to agent $j$ for each $k\in [K]$, resulting in $\mathfrak{n}\times K$ encrypted values.
Clearly, both communication and computation costs are increased drastically with the increase of $\mathfrak{n}$.
A fully homomorphic encryption such as \cite{cheon2017homomorphic} can be employed to reduce the communication since agent $i$ can encrypt $c_{p_ip_j}$ and $x_i(k)$ for all $k \in [K]$ and allow agent $j$ to evaluate $P_{j}(x_i, x_j)$ over the ciphertext; leading to $\mathfrak{n} + K + 1$ encrypted values for the whole time interval. However, this benefit comes at the expense of increased computational complexity for agent $j$ due to the high computational load of fully homomorphic schemes.
\end{remark}
\begin{remark}[Beyond polynomial functions]\label{r:gen-class}
We can privately evaluate a wider class of functions represented by $$P_{j}(x_i, x_j) = \sum\limits_{{\p},{\pj}}c_{p_ip_j}{f_i^{({\p})}(x_i)f_j^{({\pj})}(x_j)}, \qquad 
W_j=\sum\limits_{{\qj}}{c_{{\qj}}g_j^{({\qj})}(x_j)},$$
where  
$f_{i}^{(\cdot)}: \R \to \R$, and $g_{j}^{(\cdot)}: \R \to \R$. 
This can be achieved by treating $f_i^{(\cdot)}(x_i)$ as $x_i$, and $f_j^{(\cdot)}(x_j)$ and $g_{j}^{(\cdot)}(x_j)$ as $x_j$ in Algorithm \ref{Alg}. 
This generalized class of functions essentially does not introduce extra communication and computation costs since all additional computations are performed over the plain text. 
\end{remark}
For a better illustration of the protocol, we provide a simple example.
\setcounter{example}{0}
\begin{example}(cont.)
Consider again the polynomial in \eqref{e:ex1}:
\begin{equation*}\label{e:ex1}
\begin{aligned}
   \calP_{1}=
  \underbrace{2x_1^2x_2}_{P_{2}}+\underbrace{3x_1x_3}_{P_{3}} + 
  \underbrace{4x_1x_4^3}_{P_{4}} +
 \underbrace{x_1}_{W_1}\underbrace{x_2^2}_{W_2}\underbrace{(x_3^2 + 3x_3)}_{W_3}\underbrace{x_4}_{W_4}. 
\end{aligned}
\end{equation*}
For the sake of simplicity, we assume $\bld{x_j}\in {\Z_{\ge 0}}$ for $j \in \{1, 2, 3, 4\}$, otherwise an encoding-decoding scheme is used. Let node $4$ be the distinguished neighbor.
Based on Algorithm \ref{Alg},
agent $1$ generates ${\text{pk}_1}$ and ${\text{sk}_1}$ and publishes ${\text{pk}_1}$.

As for the bivariate parts, agent $1$ sends $ \ENC{2\bld{x}_{\bld{1}}^2} $ to agent $2$, $ \ENC{3\bld{x}_{\bld{1}}} $ to agent $3$, and $ \ENC{4\bld{x}_{\bld{1}}} $ to agent 4. Here, among the multiplicative terms, only $W_3$ contains privacy sensitive coefficients; hence, agent $1$ sends the encrypted values
$ \ENC{1} $ and $ \ENC{3} $ to agent $3$.

\par In the next step, agent $2$ computes $\small{\sigma_2 =   \ENC{2\bld{x}_{\bld{1}}^2}^{\bld{x}_{\bld{2}}} \ENC{\sa_2} \,\,\bmod \,\,\scrN^2}$
and $\mu_2 = {{\sm_2}\bld{x}_{\bld{2}}^{\bld{2}}}$ and sends the results to agent $1$. 
Meanwhile, agent $3$ computes the following quantities and sends them to agent $1$:
$${\sigma_3 =  \ENC{3\bld{x}_{\bld{1}}}^{\bld{x}_{\bld{3}}} \ENC{\sa_3} }\,\,\bmod \,\,\scrN^2, \quad
{ \ENC{\mu_3}  ={{  \ENC{1}  }^{{\sm_3}\bld{x}_{\bld{3}}^{\bld{2}}}{  \ENC{3}  }^{{\sm_3}\bld{x}_{\bld{3}}} \,\,\bmod \,\,\scrN^2}}.$$

\medskip{}
Next, agent $1$ computes ${{\Psi }_{1}}=\left( \sm_1{\bld{x}_{\bld{1}}} \right) {{\mu }_{2}}{{\mu }_{3}}$ and sends $ \ENC{ \Psi_1} $ to agent $4$.
\par The distinguished neighbor $4$ computes $\small{\sigma_4={{  \ENC{4\bld{x}_{\bld{1}}} }^{\bld{x}_{\bld{4}}^3}} \ENC{\sa_4} }$ and $\small{\Psi_4 ={{  \ENC{\Psi_1}  }^{{\sm_4}\bld{x}_{\bld{4}}}\,\,\bmod \,\,\scrN^2}}$, and sends back $ \sigma_4  \Psi _4\,\,\bmod \,\,\scrN^2$ to agent $1$.

\par Finally, agent $1$ decrypts $\sigma_2$, $\sigma_3$,  and $ {{\sigma }_{4}}  \Psi_4$ and aggregates them with $\sa_1$ to obtain ${{\calP}_{1}}$.\EXND
\end{example}

\begin{table*}[h]
\centering
\caption{Computational and communication complexity of the proposed protocol}\label{tab:com}
\begin{tabular}{*5c}
\toprule
Polynomial part &  \multicolumn{2}{c}{Computational complexity} & \multicolumn{2}{c}{Communication complexity}\\
\midrule
{}   & agent $i$   & $j\in \calN_i$    & agent $i$   & $j\in \calN_i$ \\

\hline
Bivariate   &   \makecell{$\mathcal{O}\big(|\calN_i|K\sigma^3d^2)\big({\mathrm{Enc}})$, \\ $\mathcal{O}\big(|\calN_i|K\sigma^3\big)({\mathrm{Dec}})$} & $\mathcal{O}\big(Kl\sigma^2d^2\big)$   & $\mathcal{O}\big( |\calN_i|K\sigma d^2\big)$  & $\mathcal{O}\big( K\sigma\big)$\\
\hline
Multivariate   &  \makecell{$\mathcal{O}\big(|\calN_i|\sigma^3dT\big)({\mathrm{Enc}})$,\\ $\mathcal{O}\big(|\calN_i|K\sigma^3 T\big)({\mathrm{Dec}})$} & $\mathcal{O}\big(Kl\sigma^2dT\big)$   & $\mathcal{O}\big( |\calN_i|\sigma Td\big)$  & $\mathcal{O}\big( K\sigma T\big)$\\
\end{tabular}
\end{table*}

{\color{blue}\subsection{Computational and communication complexity}\label{sub:complexity}
 In this subsection, we quantify the computational and communication complexity of the proposed protocol. The result 
is summarized in Table \ref{tab:com}, with $\mathcal{O}(\cdot)$ indicating how the complexity of the algorithm scales with the parameter under investigation.  As can be seen from the table, the computational and communication complexity depend on
(i) the number of neighbors of agent $i$, i.e. $|\calN_i|$; (ii) the degree of the polynomial \eqref{e:poly_extend}, i.e. $d$; (iii) the total time index $K$; (iv) parameter $\sigma$, which denotes the size of the Paillier's public key $\scrN$ in bit; (v) parameter $l$, which denotes the size of the to be encrypted message $m$ in bit; and, (vi) number of multivariate terms in \eqref{e:poly_extend}, namely $T$.
For brevity, we provide below an explanation only for the computational complexity for agent $i$ with regard to the bivariate part of the polynomial \eqref{e:poly_extend}, i.e., the upper left quantities in Table \ref{tab:com}. The other entries of the table can be explained analogously. 

\par Recall from Subsection \ref{Cryptography} that modular multiplication is used in Paillier cryptosystem.  
Encryption of an $l$-bit plaintext takes $\mathcal{O}(l\sigma^2+\sigma^3) \approx \mathcal{O}(\sigma^3) $ multiplications modulo $\scrN^2$,
a multiplication of an encrypted value with a plaintext of $l$ bits takes $\mathcal{O}(l\sigma^2 )$ multiplications and 
a decryption takes $\mathcal{O}(\sigma^3)$ multiplications modulo $\scrN$.
In addition, a polynomial of two variables with degree $d$ has $(d^2+3d+2)/2)$ terms; resulting in computational cost of order $\mathcal{O}(d^2)$. The parameters $|\calN_i|$ and $K$ affect the computational cost linearly. 
\par Therefore, the computational complexity for agent $i$ due to encryption would be $\mathcal{O}\big(|\calN_i|K\sigma^3d^2\big)$ multiplication modulo $\scrN^2$, and it is $\mathcal{O}\big(|\calN_i|K\sigma^3\big)$ multiplication modulo $\scrN$ due to decryption. }

\begin{remark}[Robustness against agent dropouts]\label{rem:drop-out}
The proposed scheme is essentially robust to dropout of an agent, say $j$, during the execution of the algorithm.
This means that 
agent $i$ is able to evaluate a new polynomial $\Tilde{\calP}_i(x_i, x_{(\calN_i\backslash j)})$ that does not include $x_j$. Note that $\tilde{\calP}_i$ can be obtained from  ${\calP}_i$ by setting $p_j=0$ in \eqref{e:poly_gen}.

To endow Algorithm \ref{Alg} with this capability,
agent $i$ notifies the neighboring agents $\calN_i \backslash j$ that agent $j$ is no longer a part of the computation.
By doing so, every agent ${h\in \ONI\backslash j}$ should merge (add or multiply) her own shares with the shares of the dropped out agent. Namely,
\begin{align*}
    & \sa_{hh} (\Tilde{k}) \equiv \sa_{hh}(\Tilde{k}) + \sa_{hj}(\Tilde{k}),
    & \sm_{hh}(\Tilde{k}) \equiv \sm_{hh}(\Tilde{k})  \sm_{hj}(\Tilde{k}),
\end{align*}
where $\tilde{k}$ denotes the time index marking the dropout of agent $j$.
Every agent ${h\in \ONI\backslash j}$ obtains $\sa_h(\Tilde{k})$ and $\sm_h(\Tilde{k})$ from \eqref{e:sec_dis_2} by using the updated shares $\sa_{hh}$ and $\sm_{hh}$, and discarding the shares $\sa_{jh}(\Tilde{k})$ and $\sm_{jh}(\Tilde{k})$ which she previously generated for agent $j$.  Clearly, the newly obtained quantities $\sa_h(\Tilde{k})$ and $\sm_h(\Tilde{k})$ satisfy \eqref{e:sec}, and can serve as the input of the algorithm from the time index $k = \Tilde{k}$ onward.
\end{remark}

\section{Privacy analysis}\label{sec:pri-ana}
In this section, we focus on privacy preserving properties of the proposed algorithm. To study such properties, we partition $\calV$ into a set of corrupt $\calV_c$ and noncorrupt agents $\calV_{nc}$, where the corrupt agents may collude with each other and the noncorrupt agents are simply honest-but-curious.  
We first discuss the privacy guarantees of Algorithm \ref{Alg} in the absence and presence of colluding agents. Then, we shift our focus to a network-level analysis with multiple queries. %
\subsection{Local privacy analysis}
First, we formally prove the privacy of Algorithm \ref{Alg} in the case of no collusion. This shows that no privacy sensitive information is leaked throughout the communications dictated by the algorithm.  
\begin{proposition}\label{prp:no_coll}
Let  $\calN_i \cap \calV_c = \varnothing$ and $|\calN_i| >1$.
Then Algorithm \ref{Alg} computes $\calP_i(\cdot)$ accurately and 
preserve privacy of $\PV_j = \{ x_j\}$ for $j \in \calN_i$  against agent $i$.
Moreover, Algorithm \ref{Alg} preserves privacy of $\PV_i^{} = \{ x_i, c_{{\p}{\pj}}, c_{\qj}^{(t)}\}$ against the set $\calN_i$.
\end{proposition}
\begin{proof}
The proof uses real and ideal world paradigm to show the correctness and privacy of the algorithm. Correctness of the algorithm follows from Assumption \ref{as:hon} and privacy follows from the security of Paillier and secret sharing schemes. See \ref{app:no_collude} for a formal proof.
\end{proof}
Privacy of the neighbors of $i$ is susceptible to the collusion of agent $i$ with other neighbors. The reason for the latter is that, unlike agent $i$ that uses encryption, other agents rely on a secret sharing scheme.
Hence, we formalize next the privacy guarantees when collusion occurs with agent $i$.
\begin{theorem}[]\label{the1}
Let 
$i \in \calV_c$ and assume that $D_i\in \calV_{nc}$. %
Then Algorithm \ref{Alg} computes $\calP_i(\cdot)$ accurately and 
protect privacy of $\PV_j^{}$ for $j \in \calN_i \cap \calV_{nc}$, if 
$$|\calN_i \cap \calV_{nc}| > 1.$$

\end{theorem}
\begin{proof}
The proof is built on Proposition \ref{prp:no_coll} and uses real and ideal world paradigm to show the correctness and privacy of the algorithm. See \ref{app:no_collude}.
\end{proof}

By Theorem \ref{the1}, privacy of the neighbors of $i$ is fully preserved as long as agent $i$ has at least two noncorrupt agents and the distinguished agent does not collude with agent $i$. We note again that if the distinguished neighbor colludes with agent $i$, the subqueries $\sum_{j\in \calN_{i}}P_{j}(x_i, x_j)$ and $Q_i$ in \eqref{e:poly_extend} can still be privately and accurately computed (see also Remark \ref{rem:3}). 
\begin{remark}\label{rem:lin_con}
In the context of (average) consensus the state of the art definition for privacy is that an adversary cannot estimate the value of $x_j$ with any accuracy (see for example the definition of privacy in \cite{ Ruan2019}).
At the first glance, it seems that privacy guarantees in Theorem \ref{the1} is not stringent enough compared to this definition. However, we argue that in the consensus type problems, the proposed method guarantees the same level of privacy that exists in the literature. 
The reason is that in the case of consensus protocols, the function $\calP_i(\cdot)$ becomes affine, i.e, $W_j^t(\cdot)=0$ for all $j\in \calV$.  Hence, as long as $i$ has at least one noncorrupt neighbor $h \ne j$, an attempt of agent $i$ to infer $x_j$ would at best lead to a linear equation of the form ${{x}_{j}}+{{x}_{h}}=b$. It is then clear that agent $i$ cannot estimate the value of $x_j$ with any accuracy, i.e., $x_j$ can belong to $(-\infty, \infty)$.
On the contrary, in the case of polynomial functions, 
the mere knowledge of the target function
$\calP_i(\cdot)$ may provide agent $i$ an idea about $x_j$; an ellipsoid being a simple example.  Finally, we recall that the distinguished neighbors become redundant in the case of affine functions as they only contribute to the computation of the multivariate polynomials in \eqref{e:poly_extend}.
\end{remark}

\subsection{Network privacy analysis}\label{sec:global_collude}
So far we have examined privacy concerns that may result from the computation of $\calP_i(\cdot)$, for some $i\in \calV$,  following Algorithm \ref{Alg}. Recall that in an interconnected network each agent aims to compute a function of her neighbors.
Analogous to Theorem \ref{the1}, we can show that the execution of Algorithm \ref{Alg} by every agent $i \in \calV$ protects privacy of $\PV_j^{}$ for $j \in \calV_{nc}$.
However, depending on the class of functions to be computed, colluding agents $\calV_c$ may be able to infer privacy sensitive variables of noncoluding agents by putting together the results of their queries and carrying out a posterior analysis.
Note that such potential privacy breach is oblivious to the employed privacy-preserving algorithm and descends directly from the problem setup, namely that each agent is computing a function $\calP_i(\cdot)$. The interest in studying such privacy considerations is to first highlight the inevitable limits in the privacy guarantees, and second to  provide the designer of the control/optimization algorithm  with valuable privacy related insights. 

\par The first observation is that if the number of noncolluding agents is greater than the number of colluding ones, namely
\[
 \left| \calV_{nc} \right|>\left| \calV_{c} \right|,
\]
then the colluding agents cannot uniquely infer the vector $\{x_j\}_{j\in \calV_{nc}}$. 
However, the above guarantee is weak in that it does not ensure privacy of a specific noncorrupt agent. Next, we investigate more closely the conditions under which privacy of a single agent is guaranteed against the collective information obtained by colluding agents across the entire network.
\par Let $|\calV_c|=n$, $|\calV_{nc}|=m$. Observe that collusion of $n$ corrupt agents results in a set of polynomial equations:
\be\label{e:nl}
\Phi (x_c,x_{nc})=b,
\ee
where $x_{nc}=\{x_i\}_{i\in \calV_{nc}}$, $x_c=\{x_i\}_{i\in \calV_{c}}$ , $b\in \R^n$, and $\Phi: \R^{(n+m)}\rightarrow \R^n$.  Here, $x_{nc}$ is the indeterminate set, whereas $b$, and $x_c$, and the polynomial functions in $\Phi$ are known to the colluding agents.  
\par For technical reasons and in order to write the results more explicitly, we assume that for each $i\in \calV_c$, %
at most one variable from the set $\{x_j : {j\in \calV_{nc}\setminus i}\}$ contributes to the product of $ W_j^t$'s in \eqref{e:poly_extend}.

Moreover, without loss of generality assume that the first $m$ agents are noncorrupt. Consequently, \eqref{e:nl} reduces to 
\be\label{e:semi-lin}
\bbm a_{11} & a_{12} & \cdots & a_{1m}\\ a_{21} & a_{22} & \cdots & a_{2m} \\ \vdots & \vdots & \ddots & \vdots \\ a_{n1}& a_{n2} &\cdots & a_{nm}\ebm 
\bbm P(x_1) \\ P(x_2) \\ \vdots\\ P(x_m)\ebm=b,
\ee
where the nonlinear map $P:\R \rightarrow \R^r$ is given by $P(\alpha ):=\left[ \begin{matrix}
   \alpha  & {{\alpha }^{2}} & \cdots  & {{\alpha }^{r}}  \\
\end{matrix} \right]^\top,\,\,\forall \alpha \in \mathbb{R}$,
and $a_{ij}\in \R^{1\times r}$. Here $r$ is the maximum degree of the polynomials in \eqref{e:nl}, in terms of the indeterminate variables  $x_{nc}$. 
\par It is illustrative to first look at the special case of affine functions, where $r=1$.
Then, solutions of \eqref{e:semi-lin} are completely characterized  by
\[
x_{nc}= x_{nc}^* + (I_m-A^+A) v, \quad v\in \R^m.
\]
where $A=[{{a}_{ij}}]$, $A^+$ denotes the Moore-Penrose inverse of $A$, and $x^*$ is the vector containing the true values of $\{x_i\}_{i\in  \calV_{nc}}$.
Consequently, the value of $x_i$ with $i \in \calV_{nc}$ is uniquely identified if and only if 
\be\label{e:cond}
e_i^\top \Pi =0, 
\ee
where $\Pi:=I_m-A^+A$ and $e_i$ is the $i$th unit vector of the standard basis in $\R^m$.
Indeed if \eqref{e:cond} holds, then $x_i=x_i^*$.  Conversely,  if \eqref{e:cond} does not hold, then $x_i$  has at least two distinct solutions $x_i^*$ and $x_i^*+ \norm{e_i^T\Pi}^2$, where the latter is obtained by setting $v=e_i$ and noting that $\Pi^2=\Pi$.
The situation for $r\geq 1$ becomes more complex and gives rise to the following result:
\begin{theorem}\label{prp:net_coll}
The private variables $\{x_i\}_{i\in \calV_{nc}}$ are uniquely identified from \eqref{e:semi-lin} if and only if 
\be\label{e:identify}
\big(\{P(x_i^*)\} + \im (e_i^\top \otimes I_r)\Pi\big) \cap \im P =  \{P(x_i^*)\},
\ee
where ``$\otimes$" denotes the Kronecker product and $\im P=\{ y\in \R^r: \exists \alpha,  y=P(\alpha) \}$. 
\end{theorem}
Note that in case $r=1$, we have $\im P=\R$ and the conditions reduces to $\im (e_i^\top \Pi)$ being zero, which is equivalent to \eqref{e:cond}.

\begin{pfof}{Theorem \ref{prp:net_coll}}
Let $y_i=P(x_i)$, and $y=\col(y_i)$, $i\in \calV_{nc} $. Then, we can equivalently rewrite \eqref{e:semi-lin} as
\bse\label{e:cy-comp}
\begin{align}
\label{cy}
Ay&= b,\\
y_i\in &\im(P), \forall i.
\end{align}
\ese
Clearly, any solution to \eqref{e:semi-lin} satisfies \eqref{e:cy-comp}. Conversely, any solution to \eqref{e:cy-comp} can be mapped back to a solution of \eqref{e:semi-lin}. 
Now, all solutions to \eqref{cy} are given by
\[
y=y^* + (I_{mr}-A^+A) v, \quad v\in \R^{mr},
\]
where $y^*=\col(y_i^*)$ with $y_i^*:=P(x_i^*).$ Looking at the $i$th block row, we find that
\[
y_i=P(x_i^*) + (e_i^\top \otimes I_r)\Pi v, \quad v\in \R^{mr},
\]
where $\Pi=I_{mr}-A^+A$.  Consequently, any solution to \eqref{e:cy-comp} satisfies 
\[
y_i \in \big(\{P(x_i^*)\} + \im (e_i^\top \otimes I_r)\Pi\big) \cap \im P.
\]
Moreover, any $y_i$ satisfying the above inclusion is a solution to \eqref{e:cy-comp}. 
We conclude that  $P(x_i)$, and thus $x_i$, is uniquely identifiable if and only if \eqref{e:identify} holds.
\end{pfof}
\section{Case study}\label{sec:sim}
We demonstrate privacy and performance of the proposed algorithm in a networked system by considering a noncooperative game as  described in subsection (\ref{Motivation}) with $\blue{N} = 30$. Each player aims to minimize a cost function given by
\begin{equation*}
  J_i(x_i, x_{-i}) = a_ix_i^2 + x_i\big(\sum_{j \in {\calN_i}}{c_{ij,01}x_j}\big) + {\prod_{j \in {\ONI}}\big({c_{j,1}x_j + c_{j,2}x_j^2}}\big),   
\end{equation*}
where $x_i$ takes value from a local admissible set $\Gamma_{i} = [0,2]$. The actions need to satisfy a global affine constraint $\sum_{j\in \calV}{x_j} \geq 1$. Moreover, we assume that the players adopt the scheme in \cite{yi2019operator} for reaching GNE (see \eqref{e:grad}) with $\tau_i = \tau$. The dynamics of player $i$ is then given by
\be\label{e:poly_game}
 x_i(k+1) = \text{proj}_{\Gamma_i}\big(x_i(k) -\tau (2a_ix_i(k) + \sum_{j \in \mathcal{N}_i} {P_{j}} +  \prod_{j \in {\ONI}}W_j -\lambda_i)\big),
\ee
where $P_{j} = c_{ij,01}x_j$, $W_i = c_{i,0} + 2c_{i,1}x_i$ and $W_j = {c_{j,1}x_j + c_{j,2}x_j^2}$
are the terms specified in \eqref{e:poly_extend} and $a_i$, $c_{ij}$, and $c_{j}$ are player $i$'s private cost function parameters, randomly picked from $\Gamma_i$ for the simulation purposes.

The aim here is to privately evaluate \eqref{e:poly_game} using Algorithm \ref{Alg}.
To this end, we set $\tau =0.01$, and choose the length of Paillier's key $\scrN$ and $\Omega$ in \eqref{Gal} equal to $1024$ and $200$ bits, respectively. We assume that player $i$ has 3 neighbors, and thus her cost function depends explicitly on decisions of those neighboring players. The computations are performed\footnote{\url{https://github.com/teimour-halizadeh/polynomial-evaluation}} using a  $2.1$ GHz Intel Core i5 processor drawing on modules from Python library \cite{PythonPaillier}. Moreover, we have evaluated player $i$'s decision trajectory using plain signals, i.e. without any privacy concerns.
\par As it can be seen from Fig. \ref{f:Node1} the trajectory of player $i$ asymptotically converges to the origin  using the proposed algorithm similar to the case where a public algorithm is used. This implies that the proposed algorithm introduces no systematic error in the computation, thereby  certifying the correctness of the scheme (see also Theorem \ref{the1}). In order to investigate the computation and communication load of the proposed protocol, we change two parameters in the algorithm: 1) the length of the Paillier's key $\scrN$ \blue{in bits ($\sigma$)} and 2) number of neighbors of the player. 
The length of $\sigma$ plays an important role in the security of the Paillier cryptosystem; generally the greater the length of $\sigma$ is the more secure the Paillier scheme becomes.  As for the change in the number of neighbors, we execute the algorithm for the case $|\mathcal{N}_i| = 9$ and $|\mathcal{N}_i| = 27$. 
The results of the aforementioned changes on the computation time per time-step of the algorithm are illustrated in 
Fig. \ref{f:Computation}. 
 \begin{figure}[H]
\begin{center}
\includegraphics[width=0.4\textwidth]{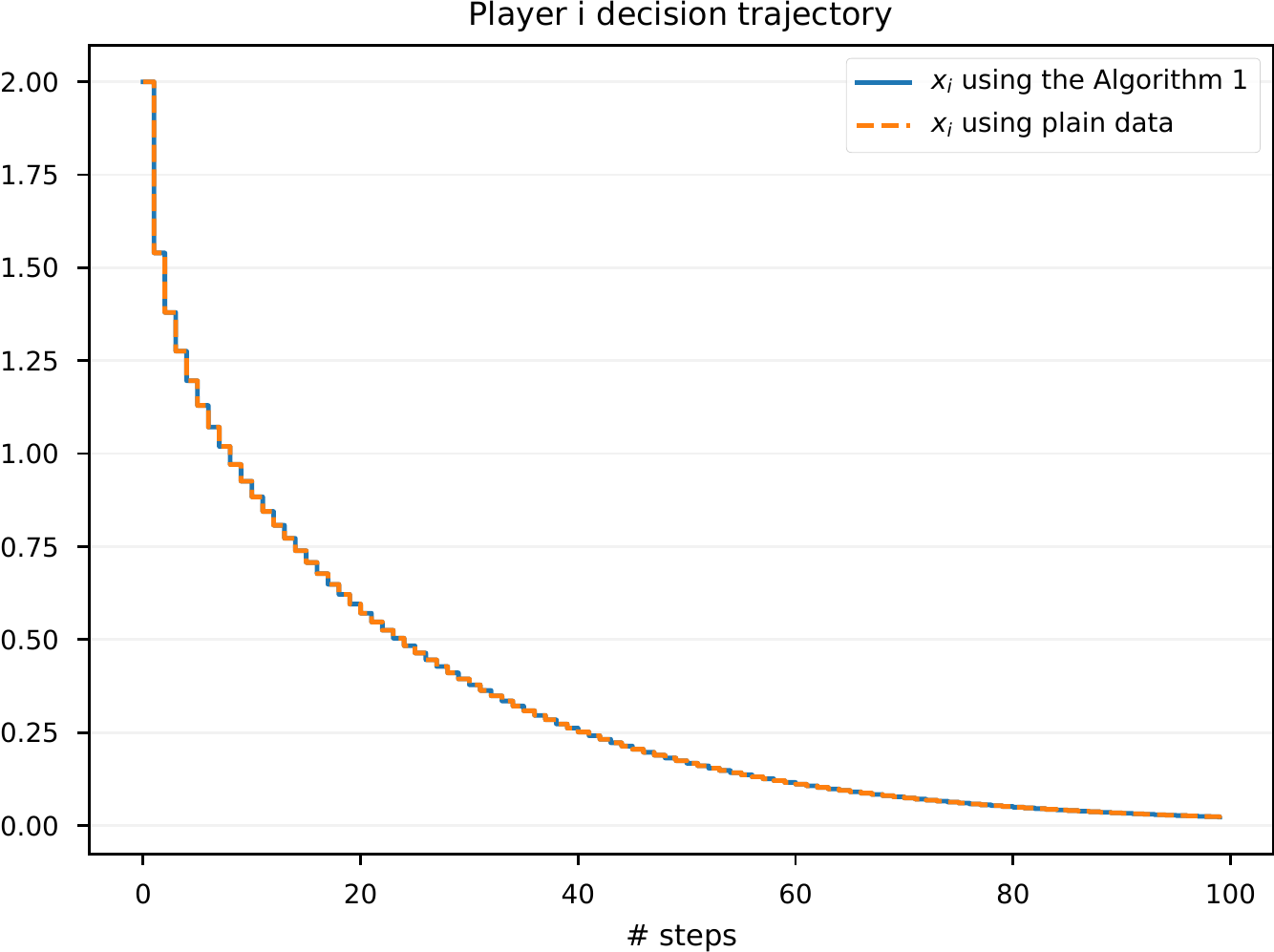}
\caption{Trajectory of player $i$ decision variable using Algorithm \ref{Alg} and plain data}\label{f:Node1}
\end{center}
\end{figure}
 \begin{figure}[H]
\begin{center}
\includegraphics[width=0.4\textwidth]{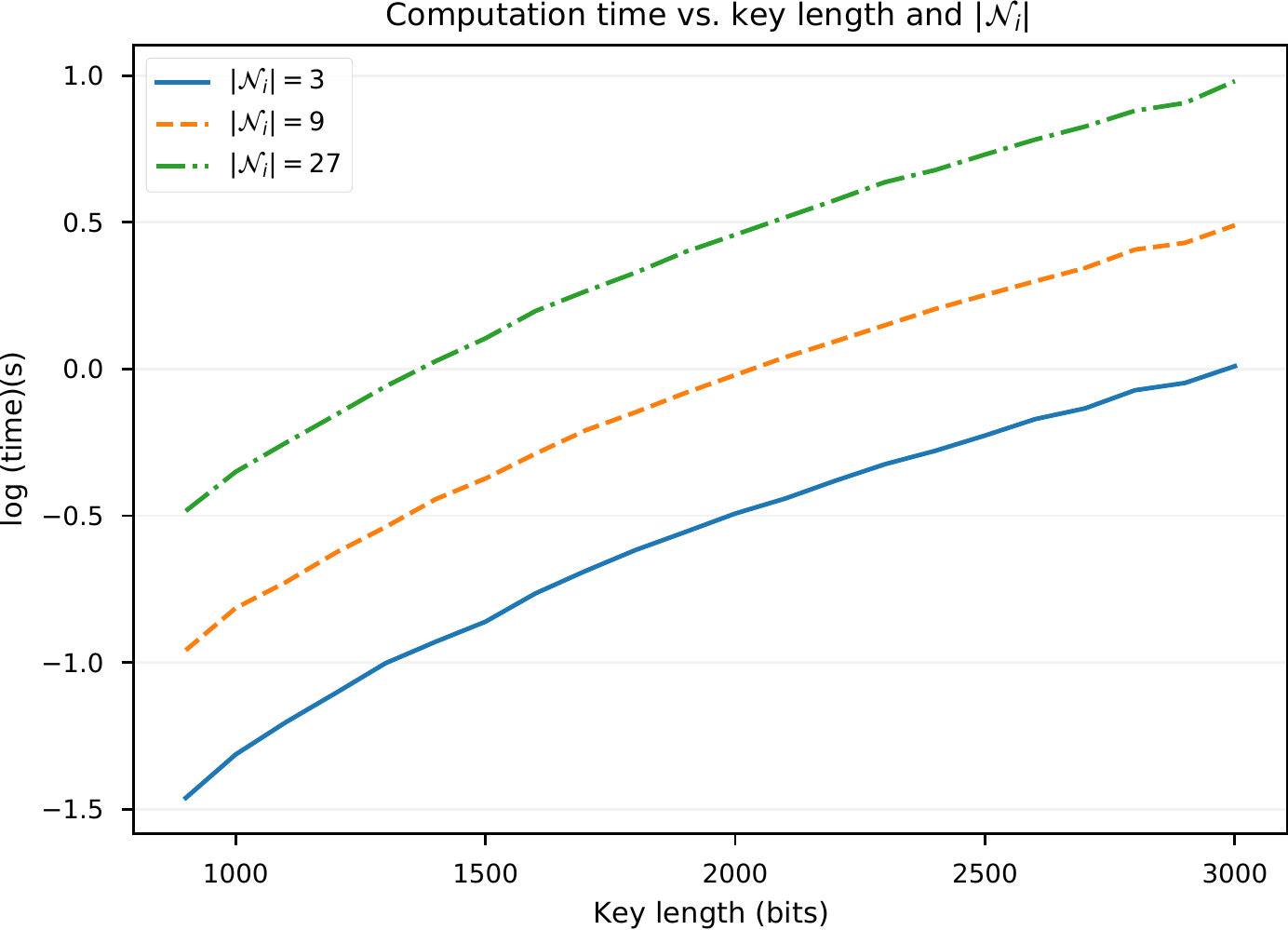}
\caption{Required computation time for the proposed algorithm with respect to the key length of Paillier scheme and number of neighbors of an agent}\label{f:Computation}
\end{center}
\end{figure}
As it is clear from this figure, the computation time increases linearly with respect to the number of neighbors, $\mathcal{O}(|\mathcal{N}_i|)$ and cubically  with respect to key length, \blue{$\calO(\sigma^3)$}.
Communication load is proportional to the size of the generated ciphertext which itself changes linearly in terms of both $|\mathcal{N}_i|$ and $\sigma$. Note that we have not employed any techniques to optimize the computation time. \blue{It is worth mentioning that the quantification results in Table \ref{tab:com} are consistent with the obtained numerical results in Fig. \ref{f:Computation}.}
\section{Conclusion}\label{sec:con}
In this study, we have presented a fully distributed algorithm for privacy preserving evaluation of a general polynomial over a network of agents. 
The algorithm is based on a suitable representation of polynomials for network systems, and adopts PHE technique and multiplicative-additive secret sharing from cryptographic tools. Furthermore, we have provided sufficient privacy-preserving conditions both at the agent and the network level.
As observed, the proposed algorithm is robust against dropout of agents, lightweight in communication and is extendable to a class of nonlinear schemes. The numerical investigations verify that the algorithm can be used to protect privacy in a network subject to additional communication and computation costs. Extensions to more general nonlinear functions and  considering possible active adversaries are among directions for future research.

\bibliography{MyReferences} 

\begin{thebibliography}{10}
\providecommand{\url}[1]{#1}
\csname url@samestyle\endcsname
\providecommand{\newblock}{\relax}
\providecommand{\bibinfo}[2]{#2}
\providecommand{\BIBentrySTDinterwordspacing}{\spaceskip=0pt\relax}
\providecommand{\BIBentryALTinterwordstretchfactor}{4}
\providecommand{\BIBentryALTinterwordspacing}{\spaceskip=\fontdimen2\font plus
\BIBentryALTinterwordstretchfactor\fontdimen3\font minus
  \fontdimen4\font\relax}
\providecommand{\BIBforeignlanguage}[2]{{%
\expandafter\ifx\csname l@#1\endcsname\relax
\typeout{** WARNING: IEEEtran.bst: No hyphenation pattern has been}%
\typeout{** loaded for the language `#1'. Using the pattern for}%
\typeout{** the default language instead.}%
\else
\language=\csname l@#1\endcsname
\fi
#2}}
\providecommand{\BIBdecl}{\relax}
\BIBdecl

\bibitem{van2019smart}
P.~Van~Aubel and E.~Poll, ``Smart metering in the netherlands: What, how, and
  why,'' \emph{International Journal of Electrical Power \& Energy Systems},
  vol. 109, pp. 719--725, 2019.

\bibitem{le2013differentially}
J.~Le~Ny and G.~J. Pappas, ``Differentially private filtering,'' \emph{IEEE
  Transactions on Automatic Control}, vol.~59, no.~2, pp. 341--354, 2013.

\bibitem{mo2016privacy}
Y.~Mo and R.~M. Murray, ``Privacy preserving average consensus,'' \emph{IEEE
  Transactions on Automatic Control}, vol.~62, no.~2, pp. 753--765, 2016.

\bibitem{nozari2017differentially}
E.~Nozari, P.~Tallapragada, and J.~Cort{\'e}s, ``Differentially private average
  consensus: Obstructions, trade-offs, and optimal algorithm design,''
  \emph{Automatica}, vol.~81, pp. 221--231, 2017.

\bibitem{nozari2016differentially}
------, ``Differentially private distributed convex optimization via functional
  perturbation,'' \emph{IEEE Transactions on Control of Network Systems},
  vol.~5, no.~1, pp. 395--408, 2016.

\bibitem{kawano2021modular}
Y.~Kawano, K.~Kashima, and M.~Cao, ``Modular control under privacy protection:
  Fundamental trade-offs,'' \emph{Automatica}, vol. 127, p. 109518, 2021.

\bibitem{kawano2020design}
Y.~Kawano and M.~Cao, ``Design of privacy-preserving dynamic controllers,''
  \emph{IEEE Transactions on Automatic Control}, vol.~65, no.~9, pp.
  3863--3878, 2020.

\bibitem{mironov2012significance}
I.~Mironov, ``On significance of the least significant bits for differential
  privacy,'' in \emph{Proceedings of the 2012 ACM Conference on Computer and
  Communications Security}, 2012, p. 650–661.

\bibitem{altafini2020system}
C.~Altafini, ``A system-theoretic framework for privacy preservation in
  continuous-time multiagent dynamics,'' \emph{Automatica}, vol. 122, p.
  109253, 2020.

\bibitem{sultangazin2020symmetries}
A.~Sultangazin and P.~Tabuada, ``Symmetries and isomorphisms for privacy in
  control over the cloud,'' \emph{IEEE Transactions on Automatic Control},
  vol.~66, no.~2, pp. 538--549, 2020.

\bibitem{monshizadeh2019plausible}
N.~Monshizadeh and P.~Tabuada, ``Plausible deniability as a notion of
  privacy,'' in \emph{IEEE 58th Conference on Decision and Control
  (CDC)}.\hskip 1em plus 0.5em minus 0.4em\relax IEEE, 2019, pp. 1710--1715.

\bibitem{Kogiso2015}
K.~Kogiso and T.~Fujita, ``Cyber-security enhancement of networked control
  systems using homomorphic encryption,'' in \emph{IEEE 54th Conference on
  Decision and Control (CDC)}.\hskip 1em plus 0.5em minus 0.4em\relax IEEE,
  2015, pp. 6836--6843.

\bibitem{Farokhi2017a}
F.~Farokhi, I.~Shames, and N.~Batterham, ``Secure and private control using
  semi-homomorphic encryption,'' \emph{Control Engineering Practice}, vol.~67,
  pp. 13--20, 2017.

\bibitem{SchulzeDarup2020}
M.~{Schulze Darup}, ``{Encrypted polynomial control based on tailored two-party
  computation},'' \emph{International Journal of Robust and Nonlinear Control},
  vol.~30, no.~11, pp. 4168--4187, 2020.

\bibitem{Kim2016}
J.~Kim, C.~Lee, H.~Shim, J.~H. Cheon, A.~Kim, M.~Kim, and Y.~Song, ``Encrypting
  controller using fully homomorphic encryption for security of cyber-physical
  systems,'' \emph{IFAC-PapersOnLine}, vol.~49, no.~22, pp. 175--180, 2016.

\bibitem{Cheon:2018}
J.~H. Cheon, K.~Han, H.~Kim, J.~Kim, and H.~Shim, ``Need for controllers having
  integer coefficients in homomorphically encrypted dynamic system,'' in
  \emph{IEEE 57th Conference on Decision and Control (CDC)}.\hskip 1em plus
  0.5em minus 0.4em\relax IEEE, 2018, pp. 5020--5025.

\bibitem{Murguia2020}
C.~Murguia, F.~Farokhi, and I.~Shames, ``{Secure and Private Implementation of
  Dynamic Controllers Using Semihomomorphic Encryption},'' \emph{IEEE
  Transactions on Automatic Control}, vol.~65, no.~9, pp. 3950--3957, 2020.

\bibitem{Alexandru2020a}
A.~B. Alexandru, K.~Gatsis, Y.~Shoukry, S.~A. Seshia, P.~Tabuada, and G.~J.
  Pappas, ``Cloud-based quadratic optimization with partially homomorphic
  encryption,'' \emph{IEEE Transactions on Automatic Control}, vol.~66, no.~5,
  pp. 2357--2364, 2021.

\bibitem{Ruan2019}
M.~{Ruan}, H.~{Gao}, and Y.~{Wang}, ``Secure and privacy-preserving
  consensus,'' \emph{IEEE Transactions on Automatic Control}, vol.~64, no.~10,
  pp. 4035--4049, 2019.

\bibitem{fang2021secure}
W.~Fang, M.~Zamani, and Z.~Chen, ``Secure and privacy preserving consensus for
  second-order systems based on paillier encryption,'' \emph{Systems \& Control
  Letters}, vol. 148, p. 104869, 2021.

\bibitem{Hadjicostis2020}
C.~N. Hadjicostis and A.~D. Dominguez-Garcia, ``{Privacy-Preserving Distributed
  Averaging via Homomorphically Encrypted Ratio Consensus},'' \emph{IEEE
  Transactions on Automatic Control}, vol.~65, no.~9, pp. 3887--3894, 2020.

\bibitem{Lu2018}
Y.~Lu and M.~Zhu, ``Privacy preserving distributed optimization using
  homomorphic encryption,'' \emph{Automatica}, vol.~96, pp. 314--325, 2018.

\bibitem{Darupcooperative}
M.~{Schulze Darup}, A.~{Redder}, and D.~E. {Quevedo}, ``Encrypted cooperative
  control based on structured feedback,'' \emph{IEEE Control Systems Letters},
  vol.~3, no.~1, pp. 37--42, 2019.

\bibitem{Alexandru2019}
A.~B. Alexandru, M.~Schulze~Darup, and G.~J. Pappas, ``Encrypted cooperative
  control revisited,'' in \emph{IEEE 58th Conference on Decision and Control
  (CDC)}.\hskip 1em plus 0.5em minus 0.4em\relax IEEE, 2019, pp. 7196--7202.

\bibitem{alexandru2021private}
A.~B. Alexandru and G.~J. Pappas, ``Private weighted sum aggregation,''
  \emph{IEEE Transactions on Control of Network Systems}, 2021.

\bibitem{Darup_Alexandru}
M.~Schulze~Darup, A.~B. Alexandru, D.~E. Quevedo, and G.~J. Pappas, ``Encrypted
  control for networked systems: An illustrative introduction and current
  challenges,'' \emph{IEEE Control Systems Magazine}, vol.~41, no.~3, pp.
  58--78, 2021.

\bibitem{paillier}
P.~Paillier, ``Public-key cryptosystems based on composite degree residuosity
  classes,'' in \emph{International conference on the theory and applications
  of cryptographic techniques}.\hskip 1em plus 0.5em minus 0.4em\relax
  Springer, 1999, pp. 223--238.

\bibitem{Hosseinalizadeh}
T.~Hossienalizadeh, F.~Turkmen, and N.~Monshizadeh, ``Private computation of
  polynomials over networks,'' in \emph{IEEE 60th Conference on Decision and
  Control (CDC)}.\hskip 1em plus 0.5em minus 0.4em\relax IEEE, 2021, pp.
  4895--4900.

\bibitem{yi2019operator}
P.~Yi and L.~Pavel, ``An operator splitting approach for distributed
  generalized {N}ash equilibria computation,'' \emph{Automatica}, vol. 102, pp.
  111--121, 2019.

\bibitem{cortes2008distributed}
\blue{Cort{\'e}s, Jorge}, ``\blue{Distributed algorithms for reaching consensus
  on general functions},'' \emph{\blue{Automatica}}, vol. \blue{44}, no.
  \blue{3}, pp. \blue{726--737}, \blue{2008}.

\bibitem{mylvaganam2017differential}
\blue{Mylvaganam, Thulasi and Sassano, Mario and Astolfi, Alessandro},
  ``\blue{A differential game approach to multi-agent collision avoidance},''
  \emph{\blue{IEEE Transactions on Automatic Control}}, vol. \blue{62}, no.
  \blue{8}, pp. \blue{4229--4235}, \blue{2017}.

\bibitem{dorfler2017gather}
\blue{D{\"o}rfler, Florian and Grammatico, Sergio},
  ``\blue{Gather-and-broadcast frequency control in power systems},''
  \emph{\blue{Automatica}}, vol. \blue{79}, pp. \blue{296--305}, \blue{2017}.

\bibitem{katz}
J.~Katz and Y.~Lindell, \emph{Introduction to modern cryptography}.\hskip 1em
  plus 0.5em minus 0.4em\relax CRC press, 2015.

\bibitem{gade2020privatizing}
S.~Gade, A.~Winnicki, and S.~Bose, ``On privatizing equilibrium computation in
  aggregate games over networks,'' \emph{IFAC-PapersOnLine}, vol.~53, no.~2,
  pp. 3272--3277, 2020.

\bibitem{shakarami2019privacy}
M.~Shakarami, C.~De~Persis, and N.~Monshizadeh, ``Distributed dynamics for
  aggregative games: Robustness and privacy guarantees,'' \emph{International
  Journal of Robust and Nonlinear Control}, vol. n/a, no. n/a.

\bibitem{dzyadyk2008theory}
V.~K. Dzyadyk and I.~A. Shevchuk, \emph{Theory of uniform approximation of
  functions by polynomials}.\hskip 1em plus 0.5em minus 0.4em\relax de Gruyter,
  2008.

\bibitem{bonawitz2017practical}
K.~Bonawitz, V.~Ivanov, B.~Kreuter, A.~Marcedone, H.~B. McMahan, S.~Patel,
  D.~Ramage, A.~Segal, and K.~Seth, ``Practical secure aggregation for
  privacy-preserving machine learning,'' in \emph{proceedings of the 2017 ACM
  SIGSAC Conference on Computer and Communications Security}, 2017, pp.
  1175--1191.

\bibitem{hardy1979introduction}
G.~H. Hardy, E.~M. Wright \emph{et~al.}, \emph{An introduction to the theory of
  numbers}.\hskip 1em plus 0.5em minus 0.4em\relax Oxford university press,
  1979.

\bibitem{lindell2017tutorials}
Y.~Lindell, \emph{Tutorials on the Foundations of Cryptography: Dedicated to
  Oded Goldreich}.\hskip 1em plus 0.5em minus 0.4em\relax Springer, 2017.

\bibitem{daemen2002design}
\blue{Daemen, Joan and Rijmen, Vincent}, \emph{\blue{The design of
  Rijndael}}.\hskip 1em plus 0.5em minus 0.4em\relax \blue{Springer},
  \blue{2002}, vol. \blue{2}.

\bibitem{cheon2017homomorphic}
J.~H. Cheon, A.~Kim, M.~Kim, and Y.~Song, ``Homomorphic encryption for
  arithmetic of approximate numbers,'' in \emph{International Conference on the
  Theory and Application of Cryptology and Information Security}.\hskip 1em
  plus 0.5em minus 0.4em\relax Springer, 2017, pp. 409--437.

\bibitem{PythonPaillier}
\BIBentryALTinterwordspacing
C.~Data61, ``Python paillier library,'' \emph{GitHub Repository}, 2013.
  [Online]. Available: \url{https://github.com/data61/python-paillier}
\BIBentrySTDinterwordspacing

\end{thebibliography}
\appendix
\section{}{\label{app:no_collude}
To provide a formal proof,  we present the definitions of view and simulator in a protocol. 

\begin{definition}[View]\cite[p. 283]{ lindell2017tutorials} Let $f(x,y) = (f_1(x,y), f_2(x,y))$ be a function, and let $\pi$ be a two party protocol(or algorithm) for computing $f$. The view of party $i$ ($i \in \{1, 2\}$) during an execution of $\pi$ on $(x, y)$ and security parameter $n$ is denoted by $\View_{i}^{\pi}(x, y, n)$ and equals $(w, r^i; m_1^i,\ldots, m_t^i)$ where $w \in \{x, y\}$, $r^i$ is the random number used by party $i$, and $m_j^i,$ represents the $j$-th message that she received. 
\end{definition}

\begin{definition}[Simulator]\cite[p. 278]{ lindell2017tutorials} Let $f(x,y) = (f_1(x,y), f_2(x,y))$ be a function, and let $\pi$ be a two party protocol for computing $f$. 
A simulator for party $i$ ($i \in \{1, 2\}$) $\Sim_i^{\pi}$ is a probabilistic polynomial-time algorithm which 
given the input and output of $i$ , $(w, f_i(x,y))$ where $w \in \{x, y\}$  can result an output whose distribution is exactly the same as $\View_{i}^{\pi}(x, y, n)$.
\end{definition}
\begin{pfof}{Proposition \ref{prp:no_coll}}
To prove this proposition, we use the simulation based paradigm also known as real/ideal world \cite[Chap. 6]{ lindell2017tutorials}.
For the deterministic function \eqref{e:poly_extend}, the security of the proposed algorithm can be shown by verifying its 1) correctness and 2) privacy. 
The proposed algorithm is correct since the agents are honest-but-curious and hence the correct value of $\calP_i(\cdot)$ is obtained by following the Protocol \ref{Alg}.
To prove privacy of  $\PV_j^{}$ for $j \in \calN_i$ against agent $i$, we need to establish the existence of a simulator $\Sim_{i}^{\pi}$ for $i$.
The input of agent $i$, meaning the information set she commits to the protocol is
$ \big\{ \{c_{p_ip_j}, c_{\qj}\}_{j \in \calN_i}, x_i, \sa_i, \sm_i, {\mathrm{pk}_i}, \ski\big\}:=\calI_{i}$
and the input of all agents involved in Algorithm \ref{Alg} is 
$ \big\{ \{c_{p_ip_j}, c_{\qj}\}_{j \in \calN_i}, \{x_j, \sa_j, \sm_j\}_{j \in \ONI}, {\mathrm{pk}_i}, \ski\big\}:=\calI$.
The {\View} of agent $i$ participating in Algorithm \ref{Alg} given the set $\calI$ is  
$\View_{i}^{\pi}(\calI) = \big\{ \calI_{i}, \{ \sigma_j, \mu_j\}_{j \in (\calN_i\backslash D_{i})} , \sigma_{D_i} + \Psi_{D_i}\big\}$, 
where $\sigma_j$, $\mu_j$ and $\sigma_{D_i} + \Psi_{D_i}$ are values received by agent $i$ in Steps $3-5$ of the proposed algorithm. Given $\calI_{i}$ and the output of the algorithm $\calP_i(\bld{x_i}, \bld{x_{\calN_i}})$ the simulator output is
$\Sim_{i}^{\pi}(\calI_{i}, \calP_i(\x_i, \x_{\calN_i})) =  \big\{ \calI_{i},\{ \hat{\sigma}_j, \hat{\mu}_j\}_{j \in (\calN_i \backslash D_{i})},\hat{\sigma}_{D_i} + \hat{\Psi}_{D_i} \big\}$.
We claim that $\View_{i}^{\pi}(\calI) \overset{c}{\equiv} \Sim_{i}^{\pi}(\calI_{i}, \calP_i(\bld{x_i}, \bld{x_{\calN_i}}))$, that is they are computationally indistinguishable.
This is true since $\Sim_{i}^{\pi}$ can pick the values $\big\{ \{ \hat{\sigma}_j, \hat{\mu}_j\}_{j \in (\calN_i \backslash D_{i})},\hat{\sigma}_{D_i} + \hat{\Psi}_{D_i} \big\}$ uniformly randomly from \eqref{Gal} with the condition that they satisfy the output of the protocol, $\calP_i(\x_i, \x_{\calN_i})$.
The $\Sim_{i}^{\pi}$ can do so since $|\calN_i|\ge 2$ and hence there exists at least two additive shares $\sa_j$ and $\sa_h$ (where $h\in \calN_{i}\backslash j$), and two multiplicative shares $\sm_j$ and $\sm_h$ to enable it to calculate $\hat{\mu}_j$ and $\hat{\sigma}_j$ and $\hat{\sigma}_{D_i} + \hat{\Psi}_{D_i}$ with the same distribution as $\mu_j$, $\sigma_j$ and $\sigma_{D_i} + \Psi_{D_i}$. Therefore, the privacy of $\PV_j^{}$ for $j \in \calN_i$ is preserved by Algorithm \ref{Alg}.
Moreover, agent $j \in \calN_i$  only receives as a private value $ m_1^i = \ENC{c_{p_ip_j}\bld{x}_{\bld{i}}^{\p}} $, $m_2^i =  \ENC{c_{{\qj}}} $($j \ne D_i$) and $m_3^i =  \ENC{c_{\qj}\Psi_{i}}  $($j = D_i$) from agent $i$(Step $2$ and $4$ of Algorithm \ref{Alg})
which are encrypted values by Paillier's scheme. Since this scheme is semantically secure and agent $j$ does not have the secret key $\ski$, agent $j$'s view is computationally indistinguishable from random numbers $\hat{m}_1^i, \hat{m}_2^i, \hat{m}_3^i \in \mathbb{Z}_{N^2}^{*}$. 
Therefore, the privacy of $\PV_i^{}$ is preserved by Algorithm \ref{Alg}.
\end{pfof}}

\begin{pfof}{Theorem \ref{the1}}
Correctness of Algorithm \ref{Alg} is similarly proved as of Proposition \ref{prp:no_coll}. 
Given the agent $i$, we need to prove the privacy of $\PV_j^{}$ for $j \in \big( \calN_i\cap \calV_{nc}\big) := \calV^{i}_{nc}$ against $\big( \ONI \cap \calV_{c}\big) := \calV^{i}_{c}$ and for that we need to establish the existence of a simulator $\Sim_{\calV^{i}_{c}}^{\pi}$. We consider the worst case scenario, i.e. $|\calV^{i}_{nc}| = 2$, meaning there are only 2 noncorrupt agents among the neighbors of agent $i$. Suppose that $\calV^{i}_{nc} = \{h, D_i\}$ where $h \ne D_i$. 
The input of colluding agents $\calV^{i}_{c}$ is $ \big\{ \{c_{p_ip_j}, c_{\qj}\}_{j \in \calN_i}, \{x_j, \sa_j, \sm_j\}_{j \in \calV^{i}_{c}}, {\mathrm{pk}_i}, \ski\big\}:=\calI_{\calV^{i}_{c}}$
and the input of parties involved in Algorithm \ref{Alg} is 
$ \big\{ \{c_{p_ip_j}, c_{{\qj}}\}_{j \in \calN_i}, \{x_j, \sa_j, \sm_j\}_{j \in \ONI}, {\mathrm{pk}_i}, \ski\big\}:=\calI$.
The {\View} of $\calV^{i}_{c}$ participating in the proposed algorithm given the set $\calI$ is $
\View_{\calV^{i}_{c}}^{\pi}(\calI) = \big\{ \calI_{\calV^{i}_{c}},  \sigma_h, \mu_h,  \sigma_{D_i} + \Psi_{D_i}\big\} $ where $\sigma_h$, $\mu_h$ and $\sigma_{D_i} + \Psi_{D_i}$ are values received by the set $\calV^{i}_{c}$ in Steps $3-5$ of the proposed algorithm.
 The simulator output is $\Sim_{\calV^{i}_{c}}^{\pi}(\calI_{\calV^{i}_{c}}, \calP_i(\bld{x_i}, \bld{x_{\calN_i}})) =  \{ \calI_{\calV^{i}_{c}}, \hat{\sigma}_h, \hat{\mu}_h,\hat{\sigma}_{D_i} + \hat{\Psi}_{D_i} \}$, given $\calI_{\calV^{i}_{c}}$ and the output of the algorithm.
The claim is $\View_{\calV^{i}_{c}}^{\pi}(\calI) \overset{c}{\equiv} \Sim_{\calV^{i}_{c}}^{\pi}(\calI_{\calV^{i}_{c}}, \calP_i(\bld{x_i}, \bld{x_{\calN_i}}))$, they are computationally indistinguishable. 
To see this, the simulator uses $\calI_{\calV^{i}_{c}}$ and $\calP_i(\x_i, \x_{\calN_i})$ to have the evaluation of $\calP_i(\bld{x_i}, \bld{x_h}, \bld{x_{D_i}}) = P_{h}(\bld{x_i}, \bld{x_h}) + P_{D_i}(\bld{x_i}, \bld{x_{D_i}}) + \xi W_h(\bld{x_h})W_{D_i}(\bld{x_{D_i}})$, where $\xi := \prod_{j \in \calV^{i}_{c}}W_j(\bld{x_j})$ is also known to the simulator.
 Then, the $\Sim_{\calV^{i}_{c}}^{\pi}$  picks $\hat{\sa}_h$, $\hat{\sa}_{D_i}$, $\hat{\sm}_h$, and $\hat{\sm}_{D_i}$ randomly from \eqref{Gal} such that \eqref{e:sec_add} and \eqref{e:sec_mul} hold. Next, it selects randomly $\bld{\hat{x}_h}$ and $\bld{\hat{x}_{D_i}}$ from \eqref{Gal} such that $\calP_i(\bld{x_i}, \bld{x_h}, \bld{x_{D_i}})$ holds. 
 Finally, the simulator outputs $\hat{\sigma}_h = P_{h}(\bld{x_i}, \bld{\hat{x}_h}) + \hat{\sa}_h$ and $\hat{\mu}_h = \hat{\sm}_h W_h(\bld{\hat{x}_h})$ for agent $h$, 
 and $\hat{\sigma}_{D_i} + \hat{\Psi}_{D_i} = P_{D_i}(\bld{x_i}, \bld{\hat{x}_{D_i}}) + \xi (\hat{\mu}_h)(\hat{\sm}_{D_i} W_h(\bld{\hat{x}_{D_i}}))+ \hat{\sa}_{D_i}$ for agent $D_i$. The set $\calV_c$ cannot differentiate between $\bld{x_h}$ and $\bld{\hat{x}_h}$ for agent $h$, and $\bld{x_{D_i}}$ and $\bld{\hat{x}_{D_i}}$ for agent $D_i$ since $\sigma_h$, $\mu_h$,  $\sigma_{D_i} + \Psi_{D_i}$ have the same distribution as $\hat{\sigma}_h$, $\hat{\mu}_h$, $\hat{\sigma}_{D_i} + \hat{\Psi}_{D_i}$. Therefore, the privacy of $\PV_j^{}$ for $j \in \calV^{i}_{nc}$ against $\calV^{i}_{c}$ is preserved.
\end{pfof}

\end{document}